\documentclass[11pt, letterpaper, onecolumn]{IEEEtran}
\usepackage[mathscr]{eucal}
\usepackage[cmex10]{amsmath}
\usepackage{epsfig,epsf,psfrag}
\usepackage{amssymb,amsmath,amsthm,amsfonts,latexsym}
\usepackage{amsmath,graphicx,bm,xcolor,url}
\usepackage[caption=false]{subfig} 
\usepackage{fixltx2e}
\usepackage{array}
\usepackage{verbatim}
\usepackage{bm}
\usepackage{algorithmic, cite}
\usepackage{algorithm}
\usepackage{verbatim}
\usepackage{textcomp}
\usepackage{mathrsfs}
\usepackage{epstopdf}

\catcode`~=11 \def\UrlSpecials{\do\~{\kern -.15em\lower .7ex\hbox{~}\kern .04em}} \catcode`~=13 

\allowdisplaybreaks[4]

\newcommand{\calA}{\mathcal{A}}
\newcommand{\calB}{\mathcal{B}}
\newcommand{\calC}{\mathcal{C}}

\newcommand{\calK}{\mathcal{K}}
\newcommand{\calL}{\mathcal{L}}
\newcommand{\calM}{\mathcal{M}}

\newcommand{\calP}{\mathcal{P}}

\newcommand{\calS}{\mathcal{S}}
\newcommand{\calT}{\mathcal{T}}

\newcommand{\calY}{\mathcal{Y}}

\newcommand{\bk}{\mathbf{k}}

\newcommand{\rmd}{\mathrm{d}}

\newcommand{\rmQ}{\mathrm{Q}}

\newcommand{\bbE}{\mathbb{E}}
\newcommand{\bbF}{\mathbb{F}}

\newcommand{\bbN}{\mathbb{N}}

\newcommand{\bbP}{\mathbb{P}}

\newcommand{\bbR}{\mathbb{R}}

\newcommand{\bbV}{\mathbb{V}}

\newcommand{\bbZ}{\mathbb{Z}}

\DeclareMathAlphabet{\mathbsf}{OT1}{cmss}{bx}{n}
\DeclareMathAlphabet{\mathssf}{OT1}{cmss}{m}{sl}

\newcommand{\rvA}{\mathsf{A}}

\newcommand{\rvB}{\mathsf{B}}

\newcommand{\rvf}{\mathsf{f}}
\newcommand{\rvF}{\mathsf{F}}

\newcommand{\rvg}{\mathsf{g}}

\DeclareSymbolFont{bsfletters}{OT1}{cmss}{bx}{n}  
\DeclareSymbolFont{ssfletters}{OT1}{cmss}{m}{n}
\DeclareMathSymbol{\bsfGamma}{0}{bsfletters}{'000}
\DeclareMathSymbol{\ssfGamma}{0}{ssfletters}{'000}
\DeclareMathSymbol{\bsfDelta}{0}{bsfletters}{'001}
\DeclareMathSymbol{\ssfDelta}{0}{ssfletters}{'001}
\DeclareMathSymbol{\bsfTheta}{0}{bsfletters}{'002}
\DeclareMathSymbol{\ssfTheta}{0}{ssfletters}{'002}
\DeclareMathSymbol{\bsfLambda}{0}{bsfletters}{'003}
\DeclareMathSymbol{\ssfLambda}{0}{ssfletters}{'003}
\DeclareMathSymbol{\bsfXi}{0}{bsfletters}{'004}
\DeclareMathSymbol{\ssfXi}{0}{ssfletters}{'004}
\DeclareMathSymbol{\bsfPi}{0}{bsfletters}{'005}
\DeclareMathSymbol{\ssfPi}{0}{ssfletters}{'005}
\DeclareMathSymbol{\bsfSigma}{0}{bsfletters}{'006}
\DeclareMathSymbol{\ssfSigma}{0}{ssfletters}{'006}
\DeclareMathSymbol{\bsfUpsilon}{0}{bsfletters}{'007}
\DeclareMathSymbol{\ssfUpsilon}{0}{ssfletters}{'007}
\DeclareMathSymbol{\bsfPhi}{0}{bsfletters}{'010}
\DeclareMathSymbol{\ssfPhi}{0}{ssfletters}{'010}
\DeclareMathSymbol{\bsfPsi}{0}{bsfletters}{'011}
\DeclareMathSymbol{\ssfPsi}{0}{ssfletters}{'011}
\DeclareMathSymbol{\bsfOmega}{0}{bsfletters}{'012}
\DeclareMathSymbol{\ssfOmega}{0}{ssfletters}{'012}

\newcommand{\tilB}{\tilde{B}}

\newcommand{\tilM}{\tilde{M}}

\newcommand{\tilu}{\tilde{u}}

\newcommand{\tilx}{\tilde{x}}

\newcommand{\barP}{\bar{P}}

\newcommand{\barX}{\bar{X}}

\newcommand{\blambda}{\bm{\lambda}}

\newcommand{\trho}{\tilde{\rho}}

\newcommand{\iid}{i.i.d.\ }

\DeclareMathOperator*{\argmax}{arg\,max}
\DeclareMathOperator*{\argmin}{arg\,min}

\DeclareMathOperator{\poly}{poly}

\newtheorem{theorem}{Theorem}

\newtheorem{proposition}[theorem]{Proposition}
\newtheorem{corollary}[theorem]{Corollary}
\newtheorem{definition}{Definition}

\newtheorem{remark}{Remark}

\newcommand{\qednew}{\nobreak \ifvmode \relax \else
      \ifdim\lastskip<1.5em \hskip-\lastskip
      \hskip1.5em plus0em minus0.5em \fi \nobreak
      \vrule height0.75em width0.5em depth0.25em\fi}

\usepackage{dsfont}

\usepackage{authblk}

\title{Unequal Message Protection: Asymptotic and Non-Asymptotic Tradeoffs}

\author[1]{Yanina Y.\ Shkel, {\em Student Member, IEEE}\thanks{This paper was presented in part at the International Symposium on Information Theory, Istanbul, Turkey, July 2013.}
}
\author[2]{Vincent Y.~F.\ Tan, {\em Member, IEEE}}
\author[3]{Stark C.\ Draper, {\em Member, IEEE}}

\affil[1]{Department of Electrical and Computer Engineering, University of Wisconsin - Madison}
\affil[2]{Department of Electrical and Computer Engineering, National University of Singapore}
\affil[3]{Department of Electrical and Computer Engineering, University of Toronto}

\newcommand{\assrtd}{\left(\{\calM_i\}_{i=1}^m, \rvf, \rvg \right)}
\newcommand{\Mepsilon}{\left( (M_i)_{i=1}^m,  (\epsilon_i)_{i=1}^m \right)}
\newcommand{\MepsilonCL}{\left( (M_{n,i})_{i=1}^{m_n},  (\epsilon_i)_{i=1}^{m_n} \right)}
\newcommand{\MepsilonMD}{\left( (M_{n,i})_{i=1}^{m_n},  (\epsilon_{n,i})_{i=1}^{m_n} \right)}
\newcommand{\Indic}[1]{\mathds{1} \left\{ #1 \right\}}

\begin{document}
\flushbottom
\maketitle

\begin{abstract} 
We study a form of unequal error protection that we term ``unequal message protection'' (UMP). The message set of a UMP code is a union of $m$ disjoint message classes. Each class has its own error protection requirement, with some classes needing better error protection than others. We analyze the tradeoff between rates of message classes and the levels of error protection of these codes. We demonstrate that there is a clear performance loss compared to homogeneous (classical) codes with equivalent parameters. This is in sharp contrast to previous literature that considers UMP codes. To obtain our results we generalize finite block length achievability and converse bounds due to Polyanskiy-Poor-Verd\'{u}. We evaluate our bounds for the binary symmetric and binary erasure channels, and analyze the asymptotic characteristic of the bounds in the fixed error and moderate deviations regimes. In addition, we consider two questions related to the practical construction of UMP codes. First, we study a ``header'' construction that prefixes the message class into a header followed by data protection using a standard homogeneous code. We show that, in general, this construction is not optimal at finite block lengths. We further demonstrate that our main UMP achievability bound can be obtained using coset codes, which suggests a path to implementation of tractable UMP codes. 
\end{abstract}
 
\section{Introduction}
\label{sec:intro}

We consider a channel coding problem of communicating a random message $w$, selected from a set of messages $\calM$, over a noisy channel $W$. Our problem is different from the classical channel coding set up in the following ways. First, we dispense with the usual assumption that messages in $\calM$ are equiprobable. Second, we consider {\em unequal error protection} (UEP), that is, some information is provided better error guarantees than other. Our main object of study is message-wise UEP codes which we term ``unequal message protection'' (UMP) codes. The message set of a UMP code is a union of $m$ disjoint message classes, $\calM = \{\calM_i\}_{i=1}^m$. Each class has its own error protection requirement, with some classes needing better error protection than others. We assume that messages within the same class are equally likely to be selected for transmission, but messages from different message classes could have different probabilities of selection. In this way, UMP codes are well suited for modeling a non-uniform prior on the message set as well as unequal error protection.

Formally, a general  channel from $\rvA$ to $\rvB$ is a stochastic kernel $W(b|a)$ satisfying $\sum_{b\in\rvB} W(b|a)=1$ for all $a\in\rvA$. Consider the following {\em one-shot} definition of a UMP code. In other words, the channel $W$ is only used once.
\begin{definition}[UMP code] \label{def:UMPCode} 
An $\Mepsilon$-UMP code for $W$ is a tuple $\assrtd$ consisting of 
\begin{enumerate}
\item $m$ disjoint message classes $\{\calM_1,\ldots, \calM_m\}$ forming the message set $\calM:=\cup_{i=1}^m\calM_i$ and satisfying $|\calM_i|=M_i$ for each $i \in \{1, 2, \dots m\}$
\item An encoder $\rvf: \calM\to\rvA$
\item A decoder $\rvg: \rvB\to\calM$
\end{enumerate}
such that for all $i \in \{1, 2, \dots m\}$, the average error probabilities for each message class satisfy
\begin{equation} 
\frac{1}{M_i}\sum_{w  \in \calM_i} W (\rvB\setminus \rvg^{-1}(w) | \rvf(w) )\le\epsilon_i. \label{eq:UMP-avg}
\end{equation}
If the maximum probability of error for each class also satisfies 
\begin{equation} 
\max_{w \in \calM_i} W (\rvB\setminus \rvg^{-1}(w) | \rvf(w) )\le\epsilon_i \label{eq:UMP-max}
\end{equation}
we refer to the code as an $\Mepsilon$-UMP code (maximum probability of error).
\end{definition}
We call a code with one class of codewords ($m=1$) a `homogeneous code'; this corresponds to the traditional channel coding framework. Paralleling~\cite{PPV2010, thesis:Polyanskiy}, a homogeneous code with $M$ codewords and average (resp.\ maximum) error probability $\epsilon$ will be referred to as an $(M,\epsilon)$-homogeneous code (average probability of error) (resp.\ (maximum  probability of error)).

To motivate the present problem we note that it is related to a number of classical problems. First, the maximum vs. average error paradigm for homogeneous codes is intimately connected to UMP codes. In channel coding with an average probability of error criterion we are concerned with one error constraint: this is immediately captured by UMP codes with one class. In channel coding with a maximum probability of error criterion we are concerned with $M$ error constraints: the error probability of each codeword. The UMP set up is a generalization of the two since it allows for error constraint of arbitrary groupings of messages. Formally, we state the following proposition. 
\begin{proposition}
\label{prop:max}
There exists an $(M, \epsilon)$-homogeneous code (average probability of error) for $W$ if and only if there exists an $\Mepsilon$-UMP code for $W$ such that $m \geq 1$, $M_i \geq M$ and $\epsilon_i \leq \epsilon$ for some $i \in \{1,2, \dots, m\}$.
Likewise, there exists an $(M, \epsilon)$-homogeneous code (maximum probability of error) for $W$ if and only if there exists an $\Mepsilon$-UMP code for $W$ such that $m\geq M$, $M_i \geq 1$ and $\epsilon_i \leq \epsilon$ for all $i \in \{1,2, \dots, m\}$.
\end{proposition}
\begin{proof}
Both assertions follow directly from Definition  \ref{def:UMPCode}.
\end{proof}

Thus, UMP codes simultaneously capture classical channel coding with an average error probability constraint and classical channel coding with the maximum probability of error constraint, as well as a whole spectrum in between.\footnote{One may note after reading Proposition~\ref{prop:max} that the notion of an $\Mepsilon$-UMP code (maximum probability of error), see~(\ref{eq:UMP-max}), is superfluous. The same object could be represented by a UMP code with $\sum_{i=1}^m M_i$ message classes, containing one codeword in each class, and having $M_i$ classes with average error probabilities $\epsilon_i$. Nevertheless, we keep the notion of a UMP code with maximum probability of error since it is conceptually and notationally convenient to do so.} In light of this observation studying fundamental limits of UMP setting is interesting from a purely theoretical perspective. 

Secondly, UMP codes can be connected to the problem of lossless joint source-channel coding by imposing a prior distribution on the message set $\calM$. In fact, message-wise UEP has appeared explicitly or implicitly in a number of works on joint source-channel coding~\cite{Csiszar1982, WIK2011,KosVer2012, FK2013, mine:STD2014:jscc}. The main distinction between the present problem and joint source-channel coding is that in the present setting the goal is to have error guarantees for all $m$ classes simultaneously, whereas in joint source-channel coding only the expected error over the whole code is studied.  Finally, we should mention that special classes of UMP codes have been used in streaming communication~\cite{Kudryashov1979, mine:NSD2013, mine:SD2010, mine:SDN2011}. We will discuss this application of UMP codes in some greater detail in Section~\ref{sec:conclusion}.

The rest of this paper is structured as follows. For the remainder of this section we present additional definitions and discussion concerning UMP codes, as well as introduce information theoretic quantities used throughout the paper. In Section~\ref{sec:problem} we review prior work and outline the main contribution of this paper. In Section~\ref{sec:finite} we prove our finite block length achievability and converse bounds. In Section~\ref{sec:dmc} we evaluate these bounds for the binary symmetric and binary erasure channels. We also present a construction based on coset codes and numerically compare the performance of our UMP bounds to the header construction that prefixes the message class into a header followed by data protection using a standard homogeneous code. In Section~\ref{sec:asymptotic} we present an asymptotic analysis of UMP codes in the fixed error and moderate deviations regimes. We end with concluding remarks in Section~\ref{sec:conclusion}.

\subsection{Additional Definitions and Notation}
When we use the term `UMP code' we refer to the triple $\assrtd$. It may be convenient also to refer to a {\em UMP codebook} which is the collection of particular codewords associated with $\assrtd$. We denote the UMP codebook by $\calC = \bigcup_{w \in \calM} \{ \rvf(w) \}$. The UMP codebook is a union of subcodebooks associated with each message class. That is, $\calC = \bigcup \calC_i$ where $\calC_i = \bigcup_{w\in \calM_i} \{ \rvf(w) \}$.

%
%

We may be interested in additional performance metrics for UMP codes. For example, we could study the overall error of the code in addition to the errors associated with each class. This is captured by notion of expected error.

\begin{definition}[Expected Error] \label{def:ExpErr}
The expected error of an $\Mepsilon$-UMP code induced by probability vector ${\bm \mu} =(\mu_1, \dots, \mu_m)$ is 
\begin{align}
\epsilon ({\bm \mu} )  = \sum_{i=1}^m \mu_i \epsilon_i. \label{eq:experr}
\end{align}
\end{definition}

We also note that the achievability bounds presented in this paper are generalizations of homogeneous bounds developed for the maximum probability of error criterion. Proposition~\ref{prop:max} suggest why adopting some achievability techniques that work for the average, but not the maximum, probability of error paradigm is challenging. If such adaptation were possible then we could derive a homogeneous bound with maximum probability of error criterion. However, we could still adopt bounds for average probability of error paradigm to bound the expected error of the code. We will take this approach in Theorems~\ref{thm:DTavg} and~\ref{thm:RCU} of Section~\ref{sec:finite}.

If $W^n$ is a sequence of channels indexed by $n$ (for example, $W^n$ is a DMC),  we may be interested in the normalized entropy of the message set assuming that the probability of selecting a message in class $i$ is $\mu_i$. We refer to this quantity as the expected rate.
\begin{definition}[Expected Rate]
\label{def.expectedRate}
The expected rate of an $\Mepsilon$-UMP code over channel $W^n$ induced by probability vector  ${\bm \mu} =(\mu_1, \dots, \mu_m)$ is 
\begin{align}
R ({\bm \mu})  =  \frac{1}{n}   \sum_{i=1}^m \mu_i \left(\log M_i - \log \mu_i\right) \label{eq:expectedRate}
\end{align}
bits per channel use.
\end{definition}

Throughout this paper $i$ will always denote the index of a class in a UMP code, $m$ the number of classes, and $n$ the channel block length. When we study asymptotic bounds for UMP codes we will consider the situation in which the number of classes scales in block length. We will denote this scaling by $m_n$.

When we present the single-shot finite block length bounds for UMP codes in Section~\ref{sec:finite} there is no scaling in $m$ and so we use the notation of $\Mepsilon$-UMP codes. For fixed error asymptotic analysis we use $\MepsilonCL$-UMP codes. We emphasize that the error probabilities are fixed, while the number of message classes is allowed to scale in $n$. For moderate deviations asymptotic analysis we let rate and error probability scale with block length and use the notation $\MepsilonMD$-UMP codes. Again, this is to emphasize that error probabilities, number of message classes, and messages class sizes, scale with $n$.

We will use sans-serif letters to indicate alphabets in single shot setting; for example, $\rvA$ will usually denote the input alphabet, and $\rvB$ will denote the output alphabet for the channel $W$. When we apply the single-shot bounds to DMCs with transition matrix $W$ and input/output alphabets $\calA$, $\calB$ we will apply them to the channel  $W^n$ and take $\rvA = \calA^n$, $\rvB = \calB^n$. Calligraphic letters will denote sets and we will use $\Indic{\calS}$ to denote the indicator function on some set $\calS$. Finally, we define output distributions $PW$ as $PW(y) = \sum_{x} P(x) W(y|x)$ and $W_x(y) =W(y|x)$.

\subsection{Information Theoretic Quantities}
 

To state our bounds we define the {\em information density} of $(X,Y)$ with joint distribution $P_{XY}$ as
\begin{align}
\imath_{X;Y}(x;y) := \log\frac{\rmd P_{Y|X=x}}{\rmd P_{Y}}(y).
\end{align}

We also define two functions that relate to hypothesis testing. Consider a random variable $Y$ defined on $\mathsf{B}$ that can take probability measure $P$ or $Q$. A randomized test between these two distributions is defined by a random transformation $P_{Z|B}: \mathsf{B} \to \{0,1\}$ where $0$ indicates that the test chooses $Q$. The best false alarm achievable among all randomized test with detection probability at least $\alpha$ is given by
\begin{align}
\beta_{\alpha}(P,Q) :=  \inf_{P_{Z|Y} : \sum_{b \in \mathsf{B}}     P_{Z|Y}(1|b) P(b)\geq \alpha} \sum_{b \in \mathsf{B}}   P_{Z|Y}(1|b)Q(b),
\end{align}
where the minimizer $P_{Z|Y}^*$ is guaranteed to be attained by the Neyman-Pearson lemma, see for example~\cite[Appendix B]{PPV2010}. 

In addition, we define a related measure of performance for the composite hypothesis test between $Q$ and a collection $\{P_{Y|X=x}\}_{x\in \rvF}$
\begin{align}
\kappa_{\tau}(\rvF,Q) := \inf_{P_{Z|Y}: \inf _{x\in \rvF} P_{Z|Y} (1|x) \geq \tau} \sum_{b \in \rvB} Q_Y(b)P_{Z|Y}(1|b) \label{eq:kappa}
\end{align}

For our asymptotic analysis we introduce the following information theoretic quantities. Denote by $\calP$ the $(|\calA| -1)$-dimensional simplex over $\bbR^{|\calA|}$ of input probability distributions. For any fixed $P\in \calP$ define: 
\begin{itemize}
\item mutual information as
\begin{align}
I(P,W) = \bbE[\imath_{X;Y}(X;Y)] = \sum_{x\in \calA, y\in \calB} P(x) W(y|x) \log\frac{W(y|x)}{PW(y)}
\end{align}
\item conditional information variance as
\begin{align}
V(P,W)= \bbE[\bbV ar(\imath_{X;Y}(X;Y)|X)] = \sum_{x\in \calA} P(x) \sum_{y\in \calB}  \left(\log\frac{W(y|x)}{PW(y)}-D(W(\cdot|x) \| PW) \right)^2,
\end{align}
\item the channel capacity as
\begin{align}
C=\max_{P\in \calP} I(P,W),
\end{align} 
\item subset of capacity achieving distributions as
\begin{align}
\Pi= \left\{P \in \calP: I(P,W)=C \right\},
\end{align}
\item maximal and minimal conditional variance as
\begin{align}
V_{\max} &= \max_{P \in \Pi} V(P,W)\\
V_{\min} &=\min_{P\in \Pi}V(P,W)
\end{align}
\item and the $\epsilon$-dispersion as
\begin{align}
V_\epsilon = \left \{\begin{array}{cc}
V_{\min}, & \epsilon<1/2 \\
V_{\max}, & \epsilon\geq1/2 \end{array} \right.
\end{align}
\item and finally information spectrum divergence as
\begin{align}
D_s^\epsilon(P\| Q ) :=\max\bigg\{ R \in\bbR: P\Big( \Big\{x: \log\frac{P(x)}{Q(x)}\le R \Big\}\Big)\le\epsilon\bigg\}.\label{eq:InfSpcDiv}
\end{align}
\end{itemize}
\section{Problem Overview}
\label{sec:problem}

\subsection{Prior Work}

Prior work on message-wise UEP has been limited to the asymptotic setting and to discrete memoryless channels (DMC). The first study was by Csisz\'{a}r~\cite{Csiszar1982} who showed that if codewords in message class $i$ are generated at rate $R_i$, then each class of codewords can have a reliability function $E(R_i)$, where $E(R)$ is the reliability function for a homogeneous ($m=1$) codebook of rate $R$.\footnote{Provided the number message classes scales sub exponentially in channel block length $n$} A similar result, that there is no apparent performance loss from several message classes being packed into the same UMP codebook, was later obtained as part of the study of error exponents for UEP schemes by Borade-Nakibo\u{g}lu-Zheng~\cite{BNZ2009}. 

The focus of this paper is on fixed error and moderate deviations asymptotic analyses, rather than analyses of large deviations setting, is in~\cite{Csiszar1982,BNZ2009}. First, consider fixing an error probability requirement for each class and study how fast corresponding rates can grow in $n$. This question has received a lot of attention in recent literature for the homogeneous case. Let $M^\ast(\epsilon, W)$ be the largest possible homogeneous code that attains error probability $\epsilon$ over an arbitrary single-shot channel $W$(cf.~\cite[Definition 2]{thesis:Polyanskiy}). Strassen~\cite{Strassen1962}, showed that for positive dispersion DMC $W$ the following holds
\begin{align}
\log M^\ast(\epsilon, W^n) =  nC -\sqrt{nV_{\epsilon}}Q^{-1}(\epsilon) + \theta(n)  \label{eq:Strassen}
\end{align}
where $Q(\cdot)$ is the tail probability of a standard normal distribution and $\theta(n) = O(\log n)$. Since then a number of works~\cite{PPV2010, thesis:Polyanskiy, TomTan2013, preprint:AltWag2013} have obtained sharper bounds on the remainder term $\theta(n)$, of which we will make use in this paper. 

Recently, Wang-Ingber-Kochman~\cite{WIK2011} derived similar fixed-error asymptotic results for the message-wise UEP problem studied here. They demonstrated that the $\epsilon$-dispersion of each class of codewords in a codebook with $m_n$ message classes matches the $\epsilon$-dispersion of each class individually, provided $m_n$ grows at most as fast as a polynomial in block length $n$. Using the notation of our paper, their result states that there is a sequence of $\MepsilonCL$-UMP codes satisfying,
\begin{align}
\log M_{n,i} = nC -\sqrt{nV_{\epsilon}}Q^{-1}(\epsilon_i) +\theta_i(n). \label{eq:WIK}
\end{align}
where $\theta_i(n)=  O(\log n)$. Just like the study of error exponents in~\cite{Csiszar1982, BNZ2009} this setting together with the assumption of polynomial (or smaller) scaling of $m_n$ does not expose any tradeoffs between different classes of a UMP code.

In the asymptotic analysis presented in~\cite{Strassen1962, thesis:Polyanskiy, TomTan2013, preprint:AltWag2013, WIK2011} the tolerated probability of error is fixed and the gap to capacity drops as $\frac{1}{\sqrt{n}}$. Another natural question to ask is what happens if the rate of a code approaches capacity, but at a slower rate than in~(\ref{eq:Strassen}). This {\em moderate deviations} behavior was studied for $m=1$ by Altu\u{g} and Wagner in~\cite{AltWag2010} for DMCs with $V_{\min}>0$ and strictly positive entries. The positive entry assumption was later relaxed by Altu\u{g}-Wagner in~\cite{preprint:AW2012}, and by Polyanksiy-Verd\'{u} in~\cite{PolVer2010}. Polyanksiy and Verd\'{u} also addressed the zero dispersion case for DMC and the additive Gaussian noise channels (AWGN). The moderate divisions results state that for positive dispersion DMC $W$, and any sequence of positive real numbers $(\rho_n)_{n\geq 1}$ such that
\begin{align}
\rho_n \to 0, \mbox{ and } n\rho_n^2 \to \infty \label{eq:intro:rho}
\end{align}
there exists a sequence of $(M_n, \epsilon_n)$-homogeneous codes over $W$ that satisfy
\begin{align}
\log M_n = nC - n\rho_n \label{eq:M_n}
\end{align}
and
\begin{align}
\limsup_{n \to \infty} \frac{1}{n\rho^2_n} \log \epsilon_n \leq -\frac{1}{2V}.
\end{align}
Conversely, for any sequence of real numbers $(\rho_n)_{n =1}^\infty$ satisfying~(\ref{eq:intro:rho}) and any sequence of $(M_n, \epsilon_n)$-codes satisfying~(\ref{eq:M_n}) it must be the case that
\begin{align}
\liminf_{n \to \infty} \frac{1}{n\rho^2_n} \log \epsilon_n \geq -\frac{1}{2V}.
\end{align}

We will call $\frac{1}{2V}$ `moderate deviations exponent'  and $\rho_n^2$ the `speed of convergence'. This result lies between the fixed error asymptotic analysis of~\cite{Strassen1962} and the large deviations analysis~\cite{book:Gallager}. To the best of the authors' knowledge, no study of UMP codes in the moderate deviations settings has been done to date.
 
\subsection{Main Results}
In this work we present a detailed analysis of UMP codes. We focus on finite block length bounds, as well as different asymptotic regimes and scaling of $m_n$ than those considered in~\cite{Csiszar1982, BNZ2009, WIK2011}. The collection of theorems presented in this work demonstrate that there is a clear performance loss in the rates of message classes and the levels of error protection compared to homogeneous codes with equivalent parameters. 

To expose the tradeoffs between different classes of messages in an UMP code we begin by first deriving finite block length bounds in Section~\ref{sec:finite}. Our approach generalizes homogeneous achievability and converse bounds due to Polyanskiy-Poor-Verd\'{u}~\cite{thesis:Polyanskiy, PPV2010}. It turns out that in the non-asymptotic regime tradeoffs are readily apparent and have a pleasing parameterization. Let $M^\ast(\epsilon, W)$ be as before and define
\begin{align}
\calL_m = \{{\blambda} = (\lambda_1, \dots, \lambda_m) : \sum_{i=1}^m \lambda_i = 1, \lambda_i \geq 0 \quad \forall i\}. \label{eq:calL_m}
\end{align}
Our bounds reveal that for any $\blambda \in \calL_m$ there is a $\Mepsilon$-UMP code that (roughly) satisfies 
\begin{align*}
M_i \leq \lambda_i M^\ast(\epsilon_i, W), \quad i\in\{1,\dots, m\}.
\end{align*}
Conversely, every UMP code must satisfy this for some $\blambda \in \calL_m$. Thus, this parameterization characterizes our achievability bounds (cf. Corollary~\ref{cor:DTbnd} and Theorem~\ref{thm:kappabeta}) and our converse bounds (cf. Theorem~\ref{thm:meta2}).

Next, in Section~\ref{sec:asymptotic} we analyze the asymptotic behavior of our bounds for DMCs, including situations in which the number of message classes scales with the channel block length. Such scalings are characterized by:
\begin{itemize}
\item a non-decreasing sequence $m_n\in \bbN$ that can scale arbitrarily in $n$, 
\item a sequence of error probabilities $(\epsilon_i)_{i=1}^\infty$ such that all error probabilities are bounded away from zero and one,
\item a doubly semi-infinite two-dimensional array $\Lambda$ parametrized by $n$ and $i$.
\end{itemize}
For any such sequence $m_n$ we define
\begin{align}
\calL = \{\Lambda : (\Lambda_{n,1}, \dots, \Lambda_{n, m_n}) \in \calL_{m_n} \quad \forall n, \mbox{ and } \Lambda_{n,i} = 0 \mbox{ if } i>m_n\} \label{eq:calL}
\end{align}
where $\Lambda_{n,i}$ is the element of $\Lambda$ in the $n$th row and $i$th column.
This set up allows us to make the following asymptotic statement (cf. Theorem~\ref{thm:asmp}).  
Any sequence of $\MepsilonCL$-UMP codes over a positive dispersion DMC $W$ must satisfy
\begin{align}
\log M_{n,i} \leq nC - \sqrt{nV} Q^{-1}(\epsilon_i)+ \theta_i(n) -   \log \frac{1}{\Lambda_{n,i}}  
\end{align}
for some $\Lambda \in \calL$ where (similar to the $m=1$ case), $\theta_i(n) = O(1)$ if $W$ is singular and symmetric and $\theta_i(n) = \frac{1}{2}\log n +O(1)$ otherwise. On the other hand, for any $\Lambda \in \calL$  there is a sequence of $\MepsilonCL$-UMP codes over a positive dispersion DMC $W$ such that
\begin{align}
\log M_{n,i} \geq nC - \sqrt{nV} Q^{-1}(\epsilon_i)  + \tilde{\theta}_i(n) -   \log \frac{1}{\Lambda_{n,i}}
\end{align} 
where $\tilde{\theta}_i(n) = O(1)$. 
Paralleling the finite block length case the performance loss of UMP codes compared to homogeneous codes with the same error probability is captured by the set $\calL$.

Finally, we analyze UMP codes in the moderate deviations regime (cf. Theorem~\ref{thm:md}). Fix $\Lambda \in \calL$ and assume that a given collection of sequences $\left ((\rho_{n,i})_{n\geq 1}\right)_{i\geq 1}$ is such that any fixed $i$ the sequence $(\rho_{n,i})_{n\geq 1}$ satisfies~(\ref{eq:intro:rho}). Then there exists a sequence of $\MepsilonMD$-UMP codes satisfying
\begin{align}
 M_{n,i} = \lfloor 2^{nC - n\rho_{n,i} } \rfloor \label{eq:intro:MDrate}
\end{align}
and
\begin{align}
\limsup_{n \to \infty} \frac{1}{n\left(\rho_{n,i} - \frac{1}{n}\log\frac{1}{\Lambda_{n,i}}\right)^2} \log \epsilon_{n,i} \leq -\frac{1}{2V}
\end{align}
for each $1\leq i \leq \infty$. Conversely, any sequence of $\MepsilonMD$-UMP codes satisfying~(\ref{eq:intro:MDrate}) must satisfy
\begin{align}
\liminf_{n \to \infty} \frac{1}{n\left(\rho_{n,i} - \frac{1}{n}\log\frac{1}{\Lambda_{n,i}}\right)^2} \log \epsilon_{n,i} \geq -\frac{1}{2V}
\end{align}
for some $\Lambda \in \calL$. In other words, each class of the UMP code has moderate deviations exponent $\frac{1}{2V_{\min}}$ and speed of convergence $\left(\rho_{n,i} - \frac{1}{n}\log\frac{1}{\Lambda_{n,i}}\right)^2$. Recall, a sequence of homogeneous codes approaching capacity at the same rate converged to the moderate deviations exponent with speed of $\rho_{n,i}^2$, and thus the loss in the moderate deviation setting is also captured by the set $\calL$.

\subsection{On Construction Of Good UMP Codes}
One may immediately observe that for a DMC the problem of constructing UMP codes has an immediate and asymptotically optimal (in terms of rate) solution. To encode a message from one of $m$ classes for transmission over a codebook of block length $n$ allocate the first $n_0$ symbols to a {\em header} that encodes the class $i \in \{1, \dots, m\}$ of the transmitted message. Allocate the remaining $n-n_0$ symbols to transmit the message $w\in\calM_i$ by using a homogeneous code. As long as $m$ grows sub-exponentially in $n$ the rate of each message class in this header-based construction can approach capacity. This is an appealing solution since it allows us to leverage existing codes as building blocks for UMP codes. 

However, as shown in Section~\ref{sec:dmc}, the header construction is suboptimal in the finite block length regime. There is simple geometric intuition for the suboptimality. The header construction is equivalent to taking the decoding space and partitioning it into separate regions, with each region used to pack codewords from one class. The more general approach taken by our Theorems~\ref{thm:DTbnd} and~\ref{thm:kappabeta} is equivalent to mixing the classes throughout the whole decoding space. This allows for a more efficient packing of the codewords in the UMP codebook. See Figure~\ref{fig:Header} for an illustration of this idea. A more formal demonstration of the suboptimality is provided in Figures~\ref{fig:BSCplot1} through \ref{fig:BECplot2} where the header code bounds are compared to UMP coding bounds for the binary symmetric channel (BSC) and the binary erasure channel (BEC). 

\begin{figure}
\centering
\includegraphics[width=5 in]{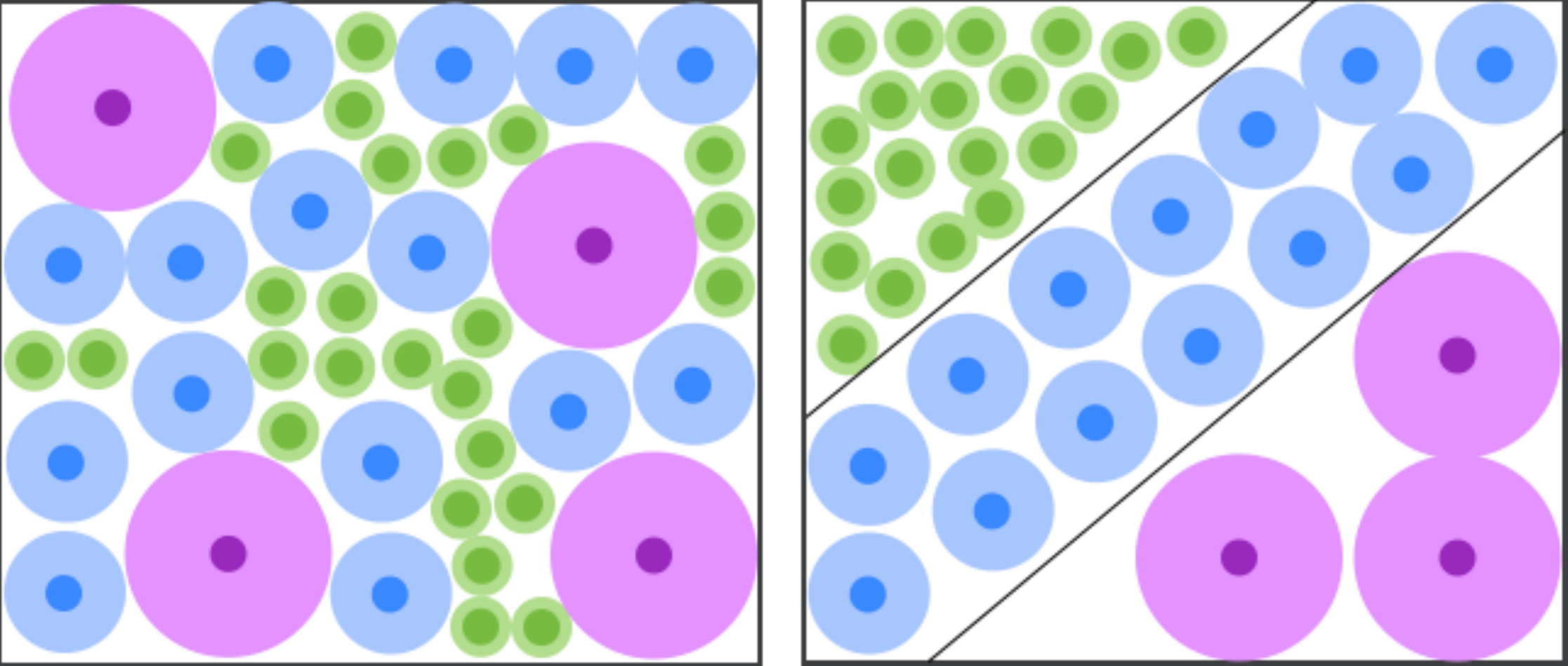}
\caption{A general UMP coding construction (left) compared with a header-based construction (right). The more general construction allows for a better packing of codewords in an UMP code.}
 \label{fig:Header}
\end{figure}

In lieu of the `header' construction we demonstrate that the performance guarantees given by Corollary~\ref{cor:DTbnd} can be achieved with a UMP code formed by taking a union of coset codes. By encoding each class with its own coset code we can construct a UMP code with good encoding complexity and decoding complexity that scales as the number of classes $m$. This result is presented for the BSC and the BEC in Theorem~\ref{thm:coset}.
\section{Finite Block Length Bounds}
\label{sec:finite}

In this section we consider an abstract channel $W$with input/output alphabets $\rvA$, $\rvB$ used once to transmit a message.

\subsection{Achievability Bounds}
We begin by extending the dependence testing (DT) for maximal probability of error bound~\cite[Theorem 21]{PPV2010} to UMP coding in Theorem~\ref{thm:DTbnd}. We follow~\cite{PPV2010} and present a compact version of the UMP DT bound in Corollary~\ref{cor:DTbnd}. Corollary~\ref{cor:DTbnd} demonstrates how the resulting family of codes is parametrized by $\calL_m$(cf.~(\ref{eq:calL_m})). In Theorem~\ref{thm:kappabeta} we extend the $\kappa \beta$-bound~\cite[Theorem 25]{PPV2010} to the UMP coding case: this extension admits the parameterization by the same $\calL_m$ as Corollary~\ref{cor:DTbnd}. Finally, a consequence of Proposition~\ref{prop:max} is that it is difficult to extend homogeneous bounds that do not work for a maximal probability of error paradigm to UMP coding. To circumvent this we make statements about the expected error of a UMP code by extending the average probability of error DT and random coding union (RCU) bounds~\cite[Theorems 16 \& 17]{PPV2010} in Theorem~\ref{thm:RCU}.

\begin{theorem}[UMP Achievability Bound]
\label{thm:DTbnd}
Let 
\begin{itemize}
\item $\calM = \bigcup_{i=1}^m \calM_i$ be a message set with $m$ disjoint message classes and $|\calM_i| = M_i$,
\item $(P_{X_i} )_{i=1}^m$ be (not necessarily distinct) distributions on $\rvA$,
\item $\left\{ \tau_i: \rvA \to [0, \infty] \right\}_{i=1}^m$ be measurable mappings,
\end{itemize}
then there exists an $\Mepsilon$-UMP code over the channel $W$ with maximum probability of error for each class not exceeding
\begin{align}
\epsilon_i &\leq \bbP \left[\imath_{X_i;Y_i}(X_i;Y_i) \leq \log\tau_i (X_i) \right] +(M_i -1) \sup_x \bbP\left[\imath_{X_i; Y_i}(x; Y_i) >  \log \tau_i (x) \right] \notag\\
&\quad + \sum_{j = 1}^{i-1} M_j \sup_x \bbP\left[\imath_{X_j; Y_j}(x; Y_i) >  \log \tau_j(x) \right] \label{eq:DTbnd}
\end{align}
where $P_{X_iY_i}(x,y)  = P_{X_i}(x) W(y|x)$ (the joint distribution induced across the channel by $P_{X_i}$) and $P_{Y_i} (y) = P_{X_i}W(y)$ (channel output distribution induced by $P_{X_i}$).
\end{theorem}
In the proof, we follow the sequential random coding technique used in~\cite[Theorem 21]{PPV2010}. In this way, we first construct the codebook for class $1$, then for class $2$, up to class $m$. The main modification is for decoding rule to vary across classes: we decode to the first codeword $c_{i,w}$ such that $\imath_{X_i;Y_i}(c_{i,w};y) > \log \tau_i(c_{i,w})$. See Appendix~\ref{appnd:finite} for the proof. By letting $m=1$ we obtain $\epsilon \leq \bbP \left[\imath_{X;Y}(X;Y) \leq \log\tau (X) \right] +(M -1) \sup_x \bbP\left[\imath_{X; Y}(x; Y) >  \log \tau (x) \right]$ which recovers~\cite[Theorem 21]{PPV2010} exactly.

Theorem~\ref{thm:DTbnd} presents bounds for the probability of error for each message class in an UMP code. By loosening these bounds we obtain the following parametrization by $\calL_m$(cf.~(\ref{eq:calL_m})).

\begin{corollary}[UMP Achievability Bound - Compact Version]
\label{cor:DTbnd}
Let $\calM$ be as in Theorem~\ref{thm:DTbnd} and suppose that the family of input distributions $(P_{X_i} )_{i=1}^m$  have the property that $P_{X_i}W = P_{X_j}W$ for all $i,j \in \{1,2, \dots, m\}$.\footnote{This holds, for example, if (i) all $m$ input distributions are the same or (ii) all $m$ input distributions are capacity achieving.} Then, for any $\blambda \in \calL_m$ there exists an $\Mepsilon$-UMP code with maximum probability of error for each class not exceeding
\begin{align}
\epsilon_i &\leq \bbP \left[\imath_{X_i; Y_i}(X_i;Y_i) \leq \log \frac{M_i}{\lambda_i} \right] + \frac{M_i}{\lambda_i} \sup_x \bbP\left[\imath_{X_i;Y_i}(x; Y_i) >  \log \frac{M_i}{\lambda_i}\right]. \label{eq:DTlambda}
\end{align}
If the CDF of $\bbP \left[\imath_{X_i;Y_i}(x; Y_i) \leq \alpha \right]$ does not depend on $x$ for any $\alpha$ we can restate~(\ref{eq:DTlambda}) as
\begin{align}
\epsilon_i \leq \bbE \left[ \exp\left\{ -\left[\imath_{X_i;Y_i}(X_i;Y_i)-\log \frac{M_i}{\lambda_i}\right]^+\right\} \right]. \label{eq:DTcompact}
\end{align} 
In~(\ref{eq:DTlambda}) and ~(\ref{eq:DTcompact}) the probability ad the expectation is taken with respect to $P_{X_iY_i}(x,y) = P_{X_i}(x)W(y|x)$.
\end{corollary}
\begin{proof}
Fix $\bm \lambda \in \calL_m$ and define
\begin{align*}
A_i = \frac{M_i}{\lambda_i}\sup_{x \in \rvA}\bbP\left[\imath_{X_i;Y_i}(x; Y_i) >  \log\tau_i(x)\right].
\end{align*}
The order in which we generate sub-codebooks for different classes in Theorem~\ref{thm:DTbnd} is arbitrary; so for a given message set, input distributions, and $\blambda \in \calL_m$ we may assume without loss of generality that $A_1 \leq A_2 \leq \dots \leq  A_m$. Observe that by loosening~(\ref{eq:DTbnd}) we obtain
\begin{align}
\epsilon_i &\leq \bbP \left[\imath_{X_i;Y_i}(X_i;Y_i) \leq \log \tau_i(x) \right] +\sum_{j = 1}^{i} M_j \sup_{x \in \rvA}\bbP\left[\imath_{X_j;Y_j}(x; Y_i) >  \log \tau_j(x)\right] \\
&= \bbP \left[\imath_{X_i;Y_i}(X_i;Y_i) \leq \log \tau_i(x)\right] +\sum_{j = 1}^{i} M_j \sup_{x \in \rvA}\bbP\left[\imath_{X_j;Y_j}(x; Y_j) >  \log \tau_j(x)\right]  \label{eq:sameY}\\
&= \bbP \left[\imath_{X_i;Y_i}(X_i;Y_i) \leq \log\tau_i(x)  \right] +\sum_{j = 1}^{i} \lambda_j A_j \\
&\leq \bbP \left[\imath_{X_i;Y_i}(X_i;Y_i) \leq \log \tau_i(x)\right] +A_i \\
&= \bbP \left[\imath_{X_i;Y_i}(X_i;Y_i) \leq \log \tau_i(x) \right] +\frac{M_i}{\lambda_i} \sup_{x \in \rvA}\bbP\left[\imath_{X_i;Y_i}(x; Y_i) >  \log\tau_i(x) \right] \label{eq:DTforMD}
\end{align}
where~(\ref{eq:sameY}) follows since $Y_i$ and $Y_j$ have the same distribution. Setting $\tau_i(x) = \frac{M_i}{\lambda_i}$ for all $x \in \rvA$ and $i \in \{1, 2, \dots, m\}$ shows~(\ref{eq:DTlambda}). To show~(\ref{eq:DTcompact}) observe that under the stated condition bound~(\ref{eq:DTlambda}) yields for any $x\in \rvA$
\begin{align}
\epsilon_i &\leq \bbP \left[\imath_{X_i;Y_i}(X_i;Y_i) \leq \log\frac{M_i}{\lambda_i}\right] + \frac{M_i}{\lambda_i}\bbP\left[\imath_{X_i;Y_i}(x; Y_i) >  \log\frac{M_i}{\lambda_i} \right]\\
&=P_{Y|X=x} \left[\imath_{X_i;Y_i}(x;Y_i) \leq \log\frac{M_i}{\lambda_i}\right] + \frac{M_i}{\lambda_i}P_{Y}\left[\imath_{X_i;Y_i}(x; Y_i) >  \log\frac{M_i}{\lambda_i} \right].
\end{align}
The result follows by repeating the argument in equations (2.129) through (2.132) of~\cite{thesis:Polyanskiy} and taking expectation  with respect to $X$ for each class.
\end{proof}

The following $\kappa \beta$-bound for UMP codes addresses the case where the codewords are constrained to belong to a subset $\rvF\subset \rvA$ for all $m$ classes. A natural extension of UMP coding to cost constraints would allow for each class to have its own cost constant $\rvF_i$. An extension of the $\kappa \beta$-bound for such a code would be interesting, and we leave it to future work. Our main motivation for presenting the bound below is to demonstrate how the same parameterization by $\calL_m$ can be applied in the case of greedy codebooks construction.
\begin{theorem}[UMP $\kappa \beta$-Bound]
\label{thm:kappabeta}
For any ${\bm \lambda} \in \calL_m$, any $\tau$ such that $0< \tau< \epsilon_i \forall i$, and any distribution $Q_Y$ on $\mathsf{B}$, there exists an $\Mepsilon$-UMP code with codewords selected from $\rvF \subset \rvA$ satisfying,
\begin{align}
M_i \geq \left \lfloor \frac{\lambda_i\kappa_\tau (\mathsf{F}, Q_Y)}{\sup_{x\in \mathsf{F}} \beta_{1-\epsilon_i+\tau}(x, Q_Y)} \right \rfloor. \label{eq:KappaBetaExt}
\end{align}
For $i=m$ we further have
\begin{align}
M_m \geq  \frac{\lambda_m\kappa_\tau (\mathsf{F}, Q_Y)}{\sup_{x\in \mathsf{F}} \beta_{1-\epsilon_m+\tau}(x, Q_Y)}. \label{eq:KappaBetaExt2}
\end{align}
\end{theorem}
The proof follows by induction. For the base case we use homogeneous $\kappa \beta$-bound~\cite[Theorem 25]{PPV2010}. For the inductive case we show that if we back off by $\lambda_i$ in the number of codewords generated in previous $m-1$ classes it is possible to add codewords to the $m$th class. See Appendix~\ref{appnd:finite} for proof. By letting $m=1$ we obtain $M \geq  \frac{\kappa_\tau (\mathsf{F}, Q_Y)}{\sup_{x\in \mathsf{F}} \beta_{1-\epsilon+\tau}(x, Q_Y)}$ which recovers~\cite[Theorem 25]{PPV2010}.

Recall that one advantage of the UMP coding framework is its ability to model a non-uniform prior on messages. To this end we study the expected error of Definition~\ref{def:ExpErr} via the following bounds. 
\begin{theorem}[Expected Error via DT-type Bound] \label{thm:DTavg}
Let
\begin{itemize}
\item $\calM = \bigcup_{i=1}^m \calM_i$ be a message set with $m$ disjoint message classes and $|\calM_i| = M_i$,
\item  $(P_{X_i} )_{i=1}^m$ be a family of distributions with the property that $P_{X_i}W = P_{X_j}W$ for all $i,j \in \{1,2, \dots, m\}$,
\item $\bm \mu$ be a probability vector of length $m$.
\end{itemize}
Then for some error vector $(\epsilon_i)_{i=1}^m$ there exists an $\Mepsilon$-UMP code with expected error induced by $\bm{\mu}$ not exceeding
\begin{align}
 \epsilon({\bm \mu}) \leq \sum_{i=1}^m \mu_i \bbE \left[ \exp\left\{ -\left[\imath_{X_i;Y_i}(X_i;Y_i)-\log \frac{M_i}{\lambda_i}\right]^+\right\} \right]  \label{eq:experrDT} 
\end{align} 
where all expectations are taken with respect to $P_{X_i Y_i} (x,y) = P_{X_i}(x) W(y|x)$.
\end{theorem}

\begin{theorem}[Expected Error via RCU-type Bound] \label{thm:RCU}
Let
\begin{itemize}
\item $\calM = \bigcup_{i=1}^m \calM_i$ be a message set with $m$ disjoint message classes and $|\calM_i| = M_i$,
\item $(P_{X_i} )_{i=1}^m$ be a familiy of (not necessarily distinct) distributions on $\rvA$,
\item $\tau_1, \dots, \tau_m \in [1, \infty]$ be $m$ real valued decoding parameters,
\item $\bm \mu$ be a probability vector of length $m$.
\end{itemize}
Then for some error vector $(\epsilon_i)_{i=1}^m$ there exists an $\Mepsilon$-UMP code with expected error induced by $\bm{\mu}$ not exceeding
\begin{align}
&\epsilon(\bm{\mu}) \leq \sum_{i=1}^m \mu_i \mathbb{E}_{X_iY_i}\left[ \min \left\{1, \sum_{j=1}^{m} \left(M_j - \Indic{i=j} \right)f_{i,j}(X_i, Y_i) \right\}\right], \label{eq:experrRCU}
\end{align}
where 
\begin{align}
&f_{i,j}(x, y)  = \mathbb{P}\left[ \left.  \imath_{X_j;Y_j}(X_j; Y_i) \geq \log \frac{\tau_i}{\tau_j} +\imath_{X_i;Y_i}(X_i;Y_i) \right| X_i=x,Y_i=y \right]
\end{align}
where $P_{X_i,Y_i} (x,y) = P_{X_i}(x) W(y|x)$ and $P_{X_j,Y_j,Y_i}(x,y,z)  = P_{X_j} (x)W(y|x)P_{X_i}W(z)$.
\end{theorem}
We follow the random coding construction of~\cite[Theorems 16 \& 17]{PPV2010}. For the DT-type bound we vary the thresholds across the different classes as in Theorem~\ref{thm:DTbnd}. For the RCU-type bound we offset the information density in class $i$ by $\log \tau_i$ and decode to the codeword with the largest modified empirical information density. Finally we apply Shannon's random coding argument after the expectation across all possible codebooks of $\epsilon({\bm \mu})$ is computed. The proof is given in Appendix~\ref{appnd:finite}. One particularly interesting choice for biasing factors is $\tau_i = \frac{M_i}{\mu_i}$. With this choice the decoding rule used to derive~(\ref{eq:experrRCU}) reduces to MAP decoding. By letting $m=1$~(\ref{eq:experrRCU}) reduces to $\epsilon \leq \mathbb{E}\left[ \min \left\{1,  (M-1)\mathbb{P}\left[ \left.  \imath_{X;Y}(\bar{X}; Y) \geq \imath_{X;Y}(X;Y) \right| X,Y \right] \right\}\right]$ which recovers~\cite[Theorems 16]{PPV2010} exactly.

\subsection{Converse Bounds}

The following is a corollary of~\cite[Theorem 26]{PPV2010}.
\begin{corollary}
\label{thm:meta1}
Consider two channels  $(\rvA,\rvB, P_{Y|X})$ and $(\rvA, \rvB, Q_{Y|X})$. Fix a UMP code with $m$ classes of messages, $\assrtd$. Let $(\epsilon_{i})_{i=1}^m$ and $(\epsilon_{i} ')_{i=1}^m$ be the respective probabilities of error for channels $P_{Y|X}$ and $Q_{Y|X}$. Let $P_{X}^i = Q_{X}^i$ be the probability distribution on $\rvA$ induced by the encoder given that a $w \in\calM_i$ was transmitted. Then we have
\begin{align}
\beta_{1-\epsilon_{i}}(P_{XY}^i, Q_{XY}^i) &\leq 1-\epsilon_{i} ' ,\quad\forall\, 1\le i\le m.
\end{align}
\end{corollary}
The result follows by appealing to~\cite[Theorem 26]{PPV2010} separately for each class of codewords.

We now apply Corollary~\ref{thm:meta1} to extend Theorem 27 in~\cite{PPV2010} to UMP codes.
\begin{theorem}
\label{thm:meta2}
Let $\calP(\rvA)$ be the space of all probability distributions on $\rvA$, and $\calP(\rvB)$ be the space of all probability distributions on $\rvB$. We can make the following statements about $\Mepsilon$-UMP codes.  For some ${\bm \lambda} \in \calL_m$ and any  $Q_{Y} \in \calP(\rvB)$,
\begin{align}
\inf_{P^i_{X}} M_i\beta_{1-\epsilon_{i}} (P^i_{XY}, P^i_X \times Q_{Y}) \leq \lambda_i \label{eq:converse1}
\end{align}
for all $1\leq i \leq m$. 
We can further restate~(\ref{eq:converse1}) as
\begin{align}
\inf_{P^1_{X}  \times \dots  \times P^m_{X}} \sup_{Q_Y} \sum_{i=1}^m M_i\beta_{1-\epsilon_{i}} (P^i_{XY}, P^i_X  \times  Q_{Y})  \leq 1 \label{eq:converse2}
\end{align}
where the $\inf$ is over the $m$-fold Cartesian product of $\mathcal{P}(\mathsf{A})$ and the $\sup$ is over $\mathcal{P}(\mathsf{B})$.
\end{theorem}
\begin{proof}
We proceed by fixing $\barP^i_X = Q^i_X$ and $Q^i_{Y|X} = Q_Y$ for an arbitrary $Q_Y$ (same for all $i$). Suppose that under this distribution $Q_Y$, the probability of decoding to a message from class $i$ is $\lambda_i$. In this case $\epsilon_{i}' = 1-\frac{\lambda_i}{M_i}$. Then we have
\begin{align}
\beta_{1-\epsilon_{i}}(\bar{P}_{XY}^i, \bar{P}^i_X \times Q_{Y}) &\leq \frac{\lambda_i}{M_i},
\end{align}
where $\barP_{XY}^i := \bar{P}_{X}^i\times P_{Y|X}$. Multiplying through by $M_i$ yields equation~(\ref{eq:converse1}). Now, adding the bounds for each class yields
\begin{align}
\sum_{i=1}^m M_i\beta_{1-\epsilon_{i}} (\barP^i_{XY}, \barP^i_X \times Q_{Y}) \leq \sum_{i=1}^m \lambda_i = 1. \label{eq.metaSum}
\end{align}
Since the above holds for all $Q_Y$ we have
\begin{align}
\sup_{Q_Y \in \calP(\rvB)} \sum_{i=1}^m M_i\beta_{1-\epsilon_{i}} (\barP^i_{XY}, \barP^i_X \times Q_{Y}) \leq 1.
\end{align}
And, since we have the freedom to choose any input distribution for each code word class
\begin{align}
\inf_{P^1_{X}\times \dots \times P^k_{X}} \sup_{Q_Y \in  \calP(\rvB)} \sum_{i=1}^m  M_i\beta_{1- \epsilon_{i}} (P^i_{XY}, P^i_X  \times  Q_{Y}) \leq 1.
\end{align}
This gives equation~(\ref{eq:converse2}). 
\end{proof}

Finally, the following result regarding constant composition codes will be useful for our asymptotic analysis. 
\begin{corollary}
\label{cor:constant}
Fix $Q_Y$ on $\rvB$ and suppose that $\beta_{\alpha} \left(P_{Y|X=x}, Q_Y\right) $ is constant for all $x\in \rvF \subset \rvB$. Then every $\Mepsilon$-UMP code with codewords belonging to $\rvF$ satisfies,
\begin{align}
M_i \leq \frac{\lambda_i}{\beta_{1-\epsilon_i} \left(P_{Y|X=x}, Q_Y\right)} \label{eq:converse3}
\end{align}
for some $\blambda \in \calL_m$ and all $1\leq i \leq m$. 
\end{corollary}
This follows directly from Theorem~\ref{thm:meta2} and~\cite[Theorem 29]{PPV2010}.

\section{Binary Symmetric and Binary Erasure Channels}
\label{sec:dmc}
In this section we evaluate the UMP bound of Corollary~\ref{cor:DTbnd} for the BSC and BEC.  The bound is evaluated for the BSC in Corollary~\ref{cor:DTBSC} and the BEC in Corollary~\ref{cor:DTBEC}. The evaluation of the converse bound of Theorem~\ref{thm:meta2} is straightforward given previous results in~\cite{PPV2010, Polaynskiy2013}. We provide it here for completeness in Corrollary~\ref{cor:metaBSC} and Corrollary~\ref{cor:metaBEC}. In Theorem~\ref{thm:coset} we show that the UMP bounds in Corollaries~\ref{cor:DTBSC} and~\ref{cor:DTBEC} can be obtained using unions of coset codes. This suggest a path to  tractable implementation of UMP codes.

We further use this section to investigate construction of UMP codes using only existing homogeneous codes. We formally state the resulting ``header bounds'' based on the homogeneous DT bound in Corollaries~\ref{cor:HeaderBSC} and~\ref{cor:HeaderBEC} and converse ``header bounds'' based on the meta converse in Corollaries~\ref{cor:HeaderConverseBSC} and~\ref{cor:HeaderConverseBEC}. Our plots in Figures~\ref{fig:BSCplot1} through~\ref{fig:BECplot2} demonstrate that, in general, the header construction is suboptimal in the finite block length regime. Specifically, the plots of the BSC (resp. BEC) of UMP bounds v.s. the header achievability bound (also based on the DT bound) provided in Figure~\ref{fig:BSCplot1} (resp. Figure~\ref{fig:BECplot1}) demonstrates that the UMP codes perform much better. When we compare UMP achievability to the header converse for the BSC in Figure~\ref{fig:BSCplot2} the results are less clear. We attribute this difference to the gap between the DT bound and the converse that is presented for homogeneous codes for the BSC. Nevertheless, for the BEC for which the gap is known to be smaller, the UMP achievability bound beats the header converse bound, cf. Figure~\ref{fig:BECplot2}. 

\subsection{Binary Symmetric Channel}
The BSC($p, n$) is the channel from $\rvA$ to $\rvB$, $\rvA = \rvB = \{0,1\}^n$, with stochastic kernel defined by
\begin{align}
W^n(y^n|x^n)= p^{|y^n - x^n|} (1-p)^{n-|y^n - x^n|}
\end{align}
where $|y^n-x^n|$ denotes the Hamming weight of the binary vector $y^n - x^n$.

\begin{corollary}[UMP Bound, BSC]\label{cor:DTBSC}
For any $\blambda \in \calL_m$, there exists an $\Mepsilon$-UMP code (maximum probability of error) for the BSC($p,n$) with
\begin{align}
\epsilon_i\! \leq \!\sum_{t=0}^n\! {n \choose t} p^t (1-p)^{n-t}\! \min \left[\!1, \!\frac{M_i}{\lambda_i}\! 2^{-n} p^{-t} (\!1 \!- \! p\!)^{t-n}\!\right]\!. \label{eq:bsc}
\end{align}
\end{corollary}
\begin{proof}
Following~\cite{PPV2010} we notice that with the equiprobable input distribution on $X^n$ the information density is $i_{X^n; Y^n}(x^n; y^n) = n \log (2-2\delta) +t \log \frac{\delta}{1-\delta}$.  The result follows by computing~(\ref{eq:DTcompact}).
\end{proof}

\begin{corollary}[Header Achievability Bound, BSC]\label{cor:HeaderBSC}
For any $0\leq n_0 \leq n$, there exists an $\Mepsilon$-UMP code for the BSC($p,n$) with
\begin{align}
\epsilon_i &\!\leq\! \sum_{t=0}^{n_0} {n_0 \choose t} p^t (1-p)^{n_0-t}\min \left[1, (m-1) 2^{-n_0-1} p ^{-t} (1-p)^{t-n_0}\right] \notag \\
&\,\,+ \sum_{t=0}^{n-n_0} {n - n_0 \choose t} p^t (1-p)^{(n - n_0)-t} \min \left[1, (M_i-1) 2^{-(n-n_0)-1} p ^{-t} (1-p)^{t-(n-n_0)}\right].
\end{align}
\end{corollary}
\begin{proof}
The result follows by applying~\cite[Theorem 34]{PPV2010} twice: once to construct a homogenous code with $m$ codewords over BSC($p, n_0$) and again to construct a homogeneous code with $M_i$ codewords over BSC($p, n-n_0$).
\end{proof}
\begin{figure}
\centering
\includegraphics[width=4.2 in]{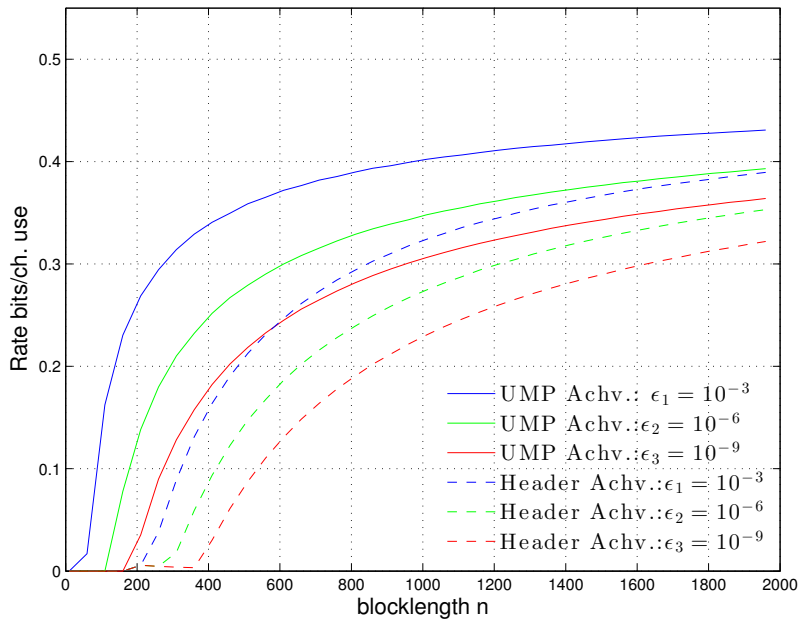}
\caption{Comparison for BSC($0.11,n$) of UMP Code in Corollary~\ref{cor:DTBSC} vs. header codes in Corollary~\ref{cor:HeaderBSC}, with $m=3$. For UMP Code the parameter $\blambda =( \frac{1}{3}, \frac{1}{3}, \frac{1}{3})$ was selected. For the header code only values of $n_0$ that can support at least one codeword in every class were considered. The best rate across all of such codes is plotted for each class.}
 \label{fig:BSCplot1}
\end{figure}

Letting $m=1$ and $n_0 = 0$ Corollary~\ref{cor:HeaderBSC} reduces to
\begin{align}
\epsilon &\leq \sum_{t=0}^{n} {n \choose t} p^t (1-p)^{n-t} \min \left[1, (M-1) 2^{-n-1} p ^{-t} (1-p)^{t-n}\right] \label{eq:homoBSC}
\end{align}
which is exactly~\cite[Theorem 34]{PPV2010}. 
Comparing~(\ref{eq:homoBSC}) and~(\ref{eq:bsc}) we can attribute the $M-1$ term being replaced by $\frac{M_i}{\lambda_i}$ to the presence of multiple classes in the code and $2^{-n-1}$ being replaced by $2^{-n}$ to the fact that we use maximum probability of error bound to obtain Corollary~\ref{cor:DTBSC}.

\begin{corollary}[UMP Converse, BSC]\label{cor:metaBSC}
Any $\Mepsilon$-UMP code over  BSC($p,n$) must satisfy
\begin{align}
M_i \leq \frac{\lambda_i}{\beta_{1-\epsilon_i}^n}, \quad \forall i = 1, \dots, m \mbox{ and some } \blambda\in\calL_m
\end{align}
where $\beta_{\alpha}^n$ is defined as
\begin{align}
\beta_\alpha^n &= (1-\rho)\beta_L + \rho\beta_{L+1} \label{eq:BECbeta}\\
\beta_l &= \sum_{j=0}^l {n \choose j} 2^{-n},
\end{align}
and where $0\leq \rho\leq 1$, and the integer $L$ are defined by
\begin{align}
\alpha &= (1-\rho)\alpha_L + \rho \alpha_{L+1}\\
\alpha_l &=\sum_{j=1}^{l-1} {n \choose j} (1-p)^{n-j}p^j.
\end{align}
\end{corollary}
\begin{proof}
Follows by identical reasoning to~\cite[Theorem 35]{PPV2010} applied to Theorem~\ref{thm:meta2}.
\end{proof}

\begin{figure}
\centering
\includegraphics[width=4.2 in]{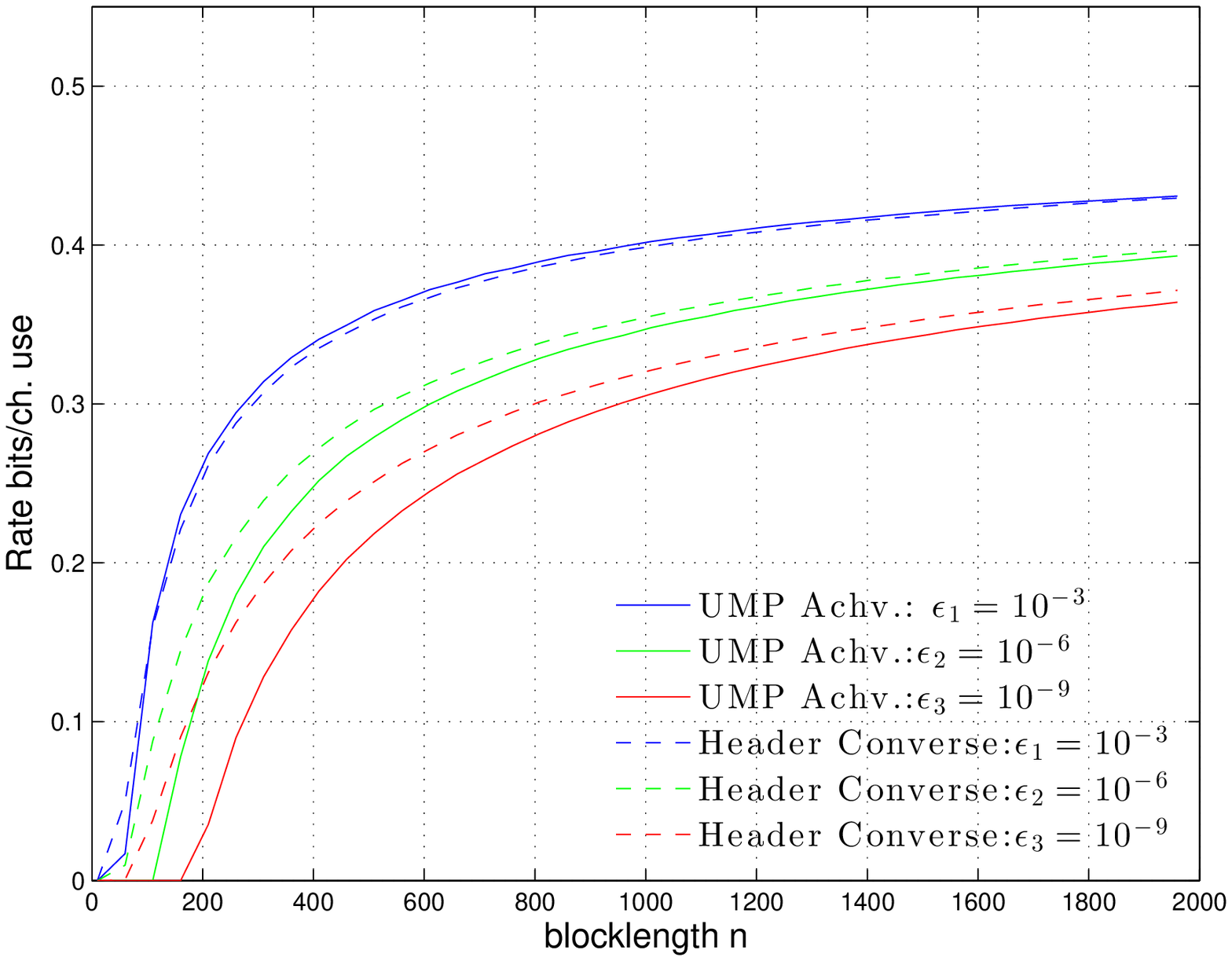}
\caption{Comparison for BSC($0.11,n$) of UMP code in Corollary~\ref{cor:DTBSC} vs. header code converse Corollary~\ref{cor:HeaderConverseBSC}), with $m=3$. For UMP code the parameter $\blambda =( \frac{1}{3}, \frac{1}{3}, \frac{1}{3})$ was selected. For the header code only values of $n_0$ which can support at least one codeword in every class were considered. The best rate across all of such codes was plotted for each class.}
 \label{fig:BSCplot2}
\end{figure}

\begin{corollary}[Header Converse Bound, BSC]\label{cor:HeaderConverseBSC}
Let $\beta_\alpha^n$ be as in~(\ref{eq:BECbeta}). Then, any $\Mepsilon$-UMP code for the BSC($p,n$) designed via the header construction must satisfy
\begin{align}
m \leq \frac{1}{\beta_{1-\epsilon_0}^{n_0}},
\end{align}
and
\begin{align}
M_i &\leq \frac{1}{\beta_{1-(\epsilon_i - \epsilon_0)}^{n-n_0}} \mbox{ if } \epsilon_i \geq \epsilon_0\\
M_i & = 0 \mbox{ otherwise,}
\end{align}
for some $0\leq n_0 \leq n$ and $0 \leq \epsilon_0 \leq 1$.
\end{corollary}
\begin{proof}
The result follows by applying~\cite[Theorem 35]{PPV2010} twice: once to construct a homogenous code with $m$ codewords over BSC($p, n_0$) and again to construct a homogeneous code with $M_i$ codewords over BSC($p, n-n_0$).
\end{proof}

\subsection{Binary Erasure Channel}

\begin{corollary}[UMP Bound, BEC]\label{cor:DTBEC}
For any ${\bm \lambda} \in \calL_m$, there exists an $\Mepsilon$-UMP code (maximum probability of error) for the  BEC($p,n$) with
\begin{align}
\epsilon_i \leq \sum_{t=0}^n {n \choose t} p^t (1-p)^{n-t} \min \left[1, \frac{M_i}{\lambda_i} 2^{t-n}\right]. \label{eq:bec}
\end{align}
\end{corollary}
\begin{proof}
Following~\cite{PPV2010} we notice that with the equiprobable input distribution on $X^n$ the information density is 
\begin{align}
i_{X^n;Y^n}(x^n; y^n) = \left\{ \begin{array}{l}
  \#\{j: y_j \neq e\}\cdot \log 2, \\ \quad \mbox{ if $y^n$ and $x^n$ agree on non-erased} \\ \quad \mbox{ positions,} \\
  -\infty, \mbox{ otherwise.}
       \end{array} \right. \nonumber
\end{align}
The result follows by computing~(\ref{eq:DTcompact}).
\end{proof}

\begin{corollary}[UMP Converse, BEC]\label{cor:metaBEC}
Any $\Mepsilon$-UMP code over  BEC($p,n$) must satisfy
\begin{align}
\epsilon_i \geq  \sum_{l=0}^n {n \choose l} p^l (1-p)^{n-l} \left(1 - \frac{\lambda_i 2^{n-l}}{M_i} \right)^+
\end{align}
for some $\blambda \in \calL_m$.
\end{corollary}
\begin{proof}
Follows by combining Theorem~\ref{thm:meta2} and~\cite[Theorem 23]{Polaynskiy2013}.
\end{proof}

\begin{corollary}[Header Achievability Bound, BEC]\label{cor:HeaderBEC}
For any $0\leq n_0 \leq n$, there exists an $\Mepsilon$-UMP code (maximum probability of error) for the  BEC($p,n$) with
\begin{align}
\epsilon_i &\leq \sum_{t=0}^{n_0} {n \choose t} p^t (1-p)^{n_0-t} \min \left[1, (m-1)2^{t-n_0-1}\right] \notag \\
&\qquad + \sum_{t=0}^{n-n_0} {(n-n_0) \choose t} p^t (1-p)^{(n-n_0)-t} \min \left[1, (M_i-1) 2^{t-(n-n_0)-1}\right].
\end{align}
\end{corollary}
\begin{proof}
The result follows by applying~\cite[Theorem 37]{PPV2010} twice: once to construct a homogenous code with $m$ codewords over BEC($p, n_0$) and again to construct a homogeneous code with $M_i$ codewords over BEC($p, n-n_0$).
\end{proof}

\begin{figure}
\centering
\includegraphics[width=4.2 in]{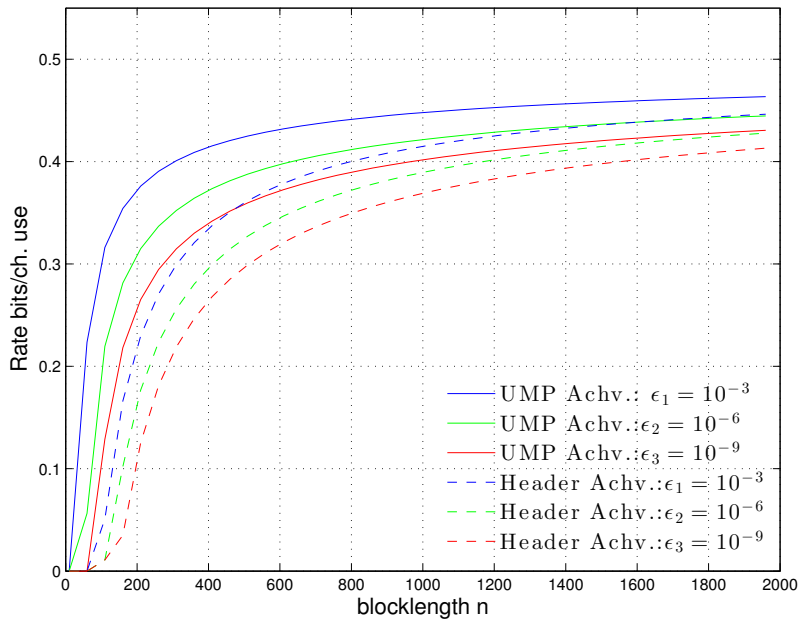}
\caption{Comparison for BEC($0.5,n$) of UMP code in Corollary~\ref{cor:DTBEC} vs. header codes in Corollary~\ref{cor:HeaderBEC}, with $m=3$. For the UMP code the parameter $\blambda =( \frac{1}{3}, \frac{1}{3}, \frac{1}{3})$ was selected. For the header code only values of $n_0$ that can support at least one codeword in every class were considered. The best rate across all of such codes is plotted for each class.}
 \label{fig:BECplot1}
\end{figure}

Letting $m=1$ and $n_0 = 0$ Corollary~\ref{cor:HeaderBEC} reduces to
\begin{align}
\epsilon &\leq \sum_{t=0}^{n} {n \choose t} p^t (1-p)^{n-t} \min \left[1, (M-1)2^{t-n-1}\right] \label{eq:homoBEC}
\end{align}
which is exactly~\cite[Theorem 37]{PPV2010}.
Comparing~(\ref{eq:homoBEC}) and~(\ref{eq:bec}) we can again attribute the $M-1$ term being replaced by $\frac{M_i}{\lambda_i}$ to the presence of multiple classes in the code and $2^{t-n-1}$ being replaced by $2^{t-n}$ to the fact that we use maximum probability of error bound to obtain Corollary~\ref{cor:DTBEC}.

\begin{corollary}[Header Converse Bound, BEC]\label{cor:HeaderConverseBEC}
Any $\Mepsilon$-UMP code over  BEC($p,n$) must satisfy
\begin{align}
\epsilon_i \geq  \sum_{l=0}^{n_0} {n_0 \choose l} p^l (1-p)^{n_0-l} \left(1 - \frac{ 2^{n_0-l}}{m} \right)^+ +  \sum_{l=0}^{n-n_0} {n-n_0 \choose l} p^l (1-p)^{(n-n_0)-l} \left(1 - \frac{2^{(n-n_0)-l}}{M_i} \right)^+
\end{align}
for some $0\leq n_0 \leq n$. 
\end{corollary}
\begin{proof}
The result follows by applying~\cite[Theorem 38]{PPV2010} twice: once to construct a homogenous code with $m$ codewords over BEC($p, n_0$) and again to construct a homogeneous code with $M_i$ codewords over BEC($p, n-n_0$).
\end{proof}

\subsection{On Achievability via Coset Codes}
In this section we address the use of coset codes to construct UMP codes. Motivated by the coset construction of~\cite{book:Gallager} we present a construction where the UMP code is a union of coset codes. This allows efficient encoding. To decode it is, in general, necessary to decode with respect to every sub-code. Thus, decoding complexity scales with the number of message classes, $m$.

\begin{figure}
\centering
\includegraphics[width=4.2 in]{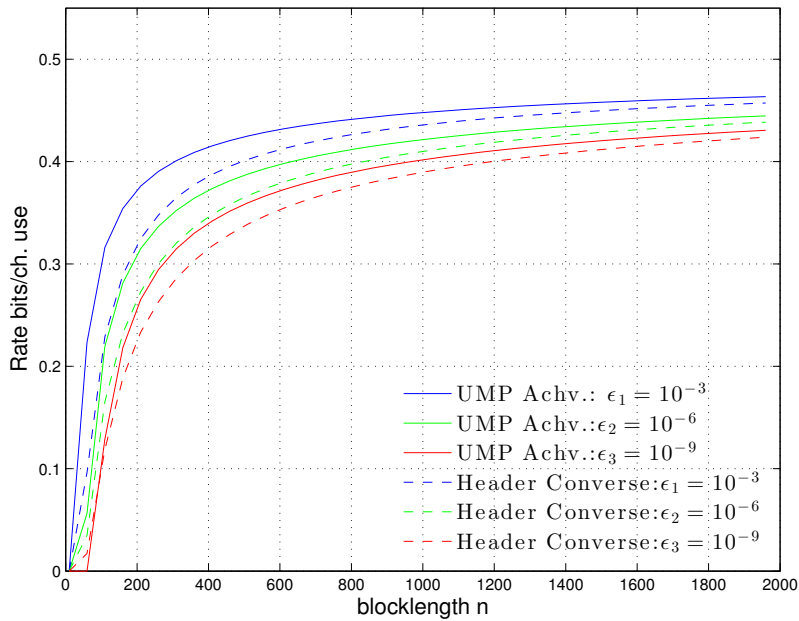}
\caption{Comparison for BEC($0.5,n$) of UMP Code in Corollary~\ref{cor:DTBEC} vs. Header Codes converse (in Corollary~\ref{cor:HeaderBEC}), with $m=3$. For UMP Code the parameter $\blambda =( \frac{1}{3}, \frac{1}{3}, \frac{1}{3})$ was selected. For the header code only values of $n_0$ which can support at least one codeword in every class were considered. The best rate across all of such codes was plotted for each class.}
 \label{fig:BECplot2}
\end{figure}

\begin{theorem}[Achievability via Coset Codes]\label{thm:coset}
Let $k_1, \dots, k_m$ be $m$ positive integers and define
\begin{align*}
M_i =  2^{k_{i}}, \quad i \in \{1, 2, \dots, m\}.
\end{align*}
Then for any $\blambda \in \calL_m$ there exists an $\Mepsilon$-UMP code (average probability of error) where $\calC = \bigcup \calC_i$ over BSC($p,n$) (respectively, BEC($p,n$)) satisfying~(\ref{eq:bsc}) (respectively,~(\ref{eq:bec})) such that each subcode $\calC_i$ is a coset of some linear code. 
\end{theorem}
\begin{proof}
We will show that under the stated conditions, we can construct $\calC_1, \calC_2, \dots, \calC_m$ such that $\calC$ satisfies~(\ref{eq:DTbnd}). The rest of the Theorem follows since~(\ref{eq:bsc}) and~(\ref{eq:bec}) can be obtained by specializing~(\ref{eq:DTbnd}) appropriately.
 
{\bf Code Construction:}
We will construct the code as follows: Let $G_i$ be a $k_i \times n$ generator matrix and $v_i$ be a $1\times n$ coset shift. Define $\calM_i := \{u_i : u_i \mbox{ is a } 1\times k_i \mbox{ binary vector}\}$. Then $\calC_i :=  \{u_iG_i + v_i : u_i \in \calM_i \}$  where multiplications and additions are over $\bbF_2$.
 
To show such code exists we sequentially generate each $(G_i, v_i)$ independently at random starting with $(G_1, v_1)$. We will show that the resulting code has good error properties and select some $(G_i, v_i)$ from the ensemble that meets the expected performance.   

{\bf Decoding Rule:}
We use a sequential threshold decoder, as in the UMP dependence testing bound, 
\begin{align}
\rvg(y^n) \!:= \!\argmin_{i,w}\left\{c_{i,w} : i_{X^n;Y^n}(c_{i,w};y^n) \!>\! \log \tau_i \right\}, 
\end{align}
where $ i \in \{1,2, \dots, m\}$, $w \in \{1,2, \dots, M_i\}$, $\tau_i = \frac{M_i}{\lambda_i}$ for all $i$, $P_{X^n, Y^n} (x^n, y^n)=  P_{X^n}(x^n)W^n(y^n|x^n)$ and $P_{X^n}$ is the uniform distribution on $\bbF_2^n$.

{\bf Error Analysis:} We will prove that the error for $\calC$ satisfies~(\ref{eq:DTbnd}) by induction on sub-codes. Consider the base case, $i=1$. We generate entries of $G_1$ and $v_1$ in an \iid manner according to a Bernoulli($\frac{1}{2}$) distribution. Let $c_{1,1} = u_1 G_1 +v_1$ be the codeword sent and $\tilu_1G_1 + v_1$ be some other codeword. The two codewords are pairwise independent and so we have that for some $G_1$ and $v_1$,
\begin{align}
\epsilon_1 &\leq \bbP \left[i_{X^n;Y^n}(X^n;Y^n) \leq \log\tau_1 \right] +(M_1 -1) \bbP\left[i_{X^n;Y^n}(\barX^n; Y^n) >  \log \tau_1 \right]
\end{align}
where $P_{\barX^n; Y^n}(x^n, y^n) = P_{X^n}(x^n) P_{X^n}W^n(y^n)$.

Now, suppose $(G_1, v_1), \dots, (G_{i-1}, v_i)$ are fixed. Generate entries of $(G_i, v_i)$ in an \iid manner according to a Bernoulli($\frac{1}{2}$) distribution. Suppose the random vector $u_iG_i+v_i$ is the true codeword sent. Then the probability that the information density of the true codeword and the output vector  is lower than the decoding threshold is bounded by, 
\begin{align}
 &\bbP \left[i_{X^n;Y^n}(u_iG_i+v_i;Y^n)  \leq \log\tau_i \right] = \bbP \left[i_{X^n;Y^n}(X^n;Y^n) \leq \log\tau_i \right]. 
\end{align}
The probability of confusing $u_i G_i + v_i$ with some other $\tilu_iG_i+v_i$ is, by pairwise independence and the uniform distribution induced,
\begin{align}
&\bbP  \left[ \left. i_{X^n;Y^n}(\tilu_iG_i+v_i;Y^n) >\log\tau_i \right| \tilu_i \neq u_i\right] = \bbP\left[i(\bar{X}^n; Y^n) >  \log \tau_i \right].
\end{align}

Finally to bound the probability of confusion with $\tilx^n = \tilu_jG_j+v_j$, a codeword in another class $j<i$, observe that $u_iG_i+v_i$ induces an equiprobable distribution on $Y^n$ and
\begin{align}
 &\bbP  \left[  i_{X^n;Y^n}(\tilx^n;Y^n) >\log\tau_j \right] \leq \sup_{x^n \in \rvA} \bbP\left[i_{X^n;Y^n}(x^n; Y^n) >  \log \tau_j \right] . \label{eq:DTLinear}
\end{align}

Since, given a random $(G_i,v_i)$ pair a codeword $u_iG_i+v_i$ satisfies these bounds for all $u_i \in \calM_i$,  the error averaged over all codewords must too. So there must exist a $(G_i, v_i)$ pair such that
\begin{align}
\epsilon_i &\leq \bbP \left[i_{X^n;Y^n}(X^n;Y^n) \leq \log\tau_i\right] +(M_i -1) \sup_x \bbP\left[i_{X^n;Y^n}(x^n; Y^n) >  \log \tau_i \right] \notag\\
& \quad+ \sum_{j = 1}^{i-1} M_j \sup_x \bbP\left[i_{X^n;Y^n}(x^n; Y^n) >  \log \tau_j \right]  
\end{align}
which shows~(\ref{eq:DTbnd}) for equiprobable $P_{X^n}$ and is sufficient to show~(\ref{eq:bsc}) for the BSC and~(\ref{eq:bec}) for the BEC.
\end{proof}

\section{Asymptotic Theorems}
\label{sec:asymptotic}
In this section we state two asymptotic theorems for the DMC. We analyze fixed error asymptotics and moderate deviations asymptotics for UMP codes and show that in both cases the performance loss compared to a homogeneous code with equivalent parameters is captured by some $\Lambda \in \calL$ (cf. equation~(\ref{eq:calL})). 

In our theorem statements we allow the number of classes $m_n$ to scale as a function of block length. One motivation for such scaling is the use of UMP codes for joint source-channel codes as in~\cite{Csiszar1982}. Note that in~\cite{Csiszar1982} the number of UMP classes needed is connected with the the number of type classes of the source. Thus, for a discrete memoryless source $m_n$ scales as a polynomial in block length. Examples of other interesting sources include~\cite{CS1996}, where the number of type classes scales exponentially in $\sqrt{n}$.

Recall that $W^n$ is a DMC with input alphabet $\calA$ and output alphabet $\calB$ if we can write,
\begin{align}
W^n(y^n|x^n)=\prod_{j=1}^n W(y_j|x_j).
\end{align}
We will apply single-shot bounds of Section~\ref{sec:finite} taking $W^n$ as the channel. We take $\rvA = \calA^n$ (respectively $\rvB = \calB^n$) to be the channel input (respectively output) alphabet. 

\begin{theorem}[Fixed Error UMP Asymptotics]
\label{thm:asmp}
Suppose that $W$ is such that $V(P^\ast_{X},W) > 0$ for all $P^\ast_X \in \Pi$. Let 
\begin{itemize}
\item $m_n$ be a sequence of class sizes (growing arbitrarily fast) in $n$,
\item $\epsilon_i$ be a sequence of error probabilities such that
\begin{align*}
\inf_{ i\in\mathbb{N}} \epsilon_i > 0 \mbox{ and } \sup_{ i\in\mathbb{N}} \epsilon_i < 1.
\end{align*}
\end{itemize}
Then, for any $\Lambda \in \calL$ there is a sequence of $\left( (M_{n,i})_{i=1}^{m_n},  (\epsilon_i)_{i=1}^{m_n} \right)$-UMP codes such that
\begin{align}
\log M_{n,i} \geq nC - \sqrt{nV_{\epsilon_i}} Q^{-1}(\epsilon_i) +  \theta_i(n)  - \log \frac{1}{\Lambda_{n,i}}.
\end{align}
Conversely, any sequence of $\left( (M_{n,i})_{i=1}^{m_n},  (\epsilon_i)_{i=1}^{m_n} \right)$-UMP codes must satisfy
\begin{align}
\log M_{n,i} \leq nC - \sqrt{nV_{\epsilon_i}} Q^{-1}(\epsilon_i)  + \tilde{\theta}_i(n) -  \log \frac{1}{\Lambda_{n,i} }
\end{align}
for some $\Lambda \in \calL$. 

The remainder terms $\theta_i(n)$ and $\tilde{\theta}_i$ satisfy
\begin{align}
K(\underline{e}, \overline{e},W) \leq \theta_i(n) \mbox{ and } \tilde{\theta}_i \leq \frac{1}{2} \log n +\bar{K}(\underline{e}, \overline{e},W)  \label{eq:remainder}
\end{align}
where $\underline{e} = \inf_{ i\in\mathbb{N}} \epsilon_i$, $\overline{e} = \sup_{ i\in\mathbb{N}} \epsilon_i$, and $K(\underline{e}, \overline{e},W), \bar{K}(\underline{e}, \overline{e},W)$ are constants which depend on $\underline{e}$, $\overline{e}$, and $W$.

If $W$ is symmetric and singular in the sense of~\cite{preprint:AltWag2013} the remainder terms for the achievability statement further satisfy
\begin{align}
 \theta_i(n)  \leq \bar{K}(\underline{e}, \overline{e},W).  \label{eq:singular}
\end{align}
\end{theorem}
The proof outline is as follows. We follow the approach of~\cite[Theorem 45]{PPV2010}. To show achievability we use the UMP achievability bound of Theorem~\ref{thm:DTbnd} and bound each term in~(\ref{eq:DTbnd}) using the Berry-Esseen theorem. The converse follows by using Theorem~\ref{thm:meta2} together with the approach of Tomamichel and Tan~\cite{TomTan2013} to obtain~(\ref{eq:remainder}) and the approach of Altu\u{g} and Wagner~\cite{preprint:AltWag2013} to obtain~(\ref{eq:singular}). See Appendix~\ref{appnd:asymptotic} for proof.

\begin{remark}
For $m=1$ Theorem~\ref{thm:asmp} reduces to the best results known in literature for most DMCs. A notable exception is the achievability bound when $W$ is non-singular for which~\cite{thesis:Polyanskiy} showed using the RCU bound that
\begin{align}
\frac{1}{2} \log n + K(\underline{e}, \overline{e},W) \leq \theta_i(n). 
\end{align}
This extension is not possible in our case due to the previously mentioned difficulty of extending the RCU bound to the framework of UMP codes.
\end{remark}

For a general DMC and $m$ growing faster than $\poly(n)$, there is a tradeoff in the sizes of different message classes of a UMP code. Two particularly interesting regimes are $m$ growing exponentially in $\sqrt{n}$ and $m$ growing exponentially in $n$. In these two regimes the tradeoffs are in the dispersion and capacity terms (respectively). 

For a symmetric singular DMC and $m$ growing as a function of $n$ there is a tradeoff in the sizes of different message classes of a UMP code. A particular regime of interest is $m_n=\poly(n)$ where the tradeoffs become apparent in the third-order $O(\log n)$ term. For $m$ constant no meaningful results can be proved since the current normal approximations do not quantify the constant term even for homogeneous codes. 

To state our next result we define a number of regularity conditions on two positive sequences $\left((\rho_n)_{n=1}^\infty, (\lambda_n)_{n=1}^\infty\right)$.
\begin{enumerate}
\item The ``homogenous'' moderate deviations condition is satisfied if
\begin{align}
\rho_{n} \to 0, \mbox{ and } n\rho_{n}^2 \to \infty  \label{eq:rho}
\end{align}
\item The ``positivity conditions'' is satisfied if
\begin{align}
(\rho_{n} - \lambda_n) >0 \label{eq:positive}
\end{align}
for all $n$ sufficiently large.
\item The ``speed of convergence'' condition is satisfied if
\begin{align}
\liminf_{n\to \infty}n(\rho_{n} - \lambda_n)^2 = \infty. \label{eq:infinity}
\end{align}
\end{enumerate}
Note, the fact that $\lambda_n > 0$  for $n$ sufficiently large together with homogeneous and positivity conditions imply $(\rho_n - \lambda_n) \to 0$. Sequences that satisfy all three of these conditions are said to satisfy moderate deviations regularity conditions.

\begin{theorem}[Moderate Deviations UMP Asymptotics]
\label{thm:md}
Suppose that $W$ is such that $V(P^\ast_{X},W) > 0$ for all $P^\ast_X \in \Pi$. Fix $\Lambda \in \calL$ and a collections of sequences $\left( (\rho_{n,i})_{n=1}^\infty \right)_{i=1}^\infty$ such that for each $i$
the pair of sequences $\left((\rho_{n,i})_{n=1}^\infty, \left(\frac{1}{n}\log\frac{1}{\Lambda_{n,i}}\right)_{n=1}^\infty \right)$ satisfy moderate deviations regularity conditions. Then, there exists a sequence of $\MepsilonMD$-UMP codes satisfying
\begin{align}
 M_{n,i} = \lfloor 2^{nC - n\rho_{n,i}} \rfloor \label{eq:MDrate}
\end{align}
and
\begin{align}
\limsup_{n \to \infty} \frac{1}{n(\rho_{n,i} - \frac{1}{n} \log \frac{1}{\Lambda_{n,i}})^2} \log \epsilon_{n,i} \leq -\frac{1}{2V_{\min}}.
\end{align}

Conversely, consider a sequence of $\MepsilonMD$-UMP codes satisfying ~(\ref{eq:rho}) and~(\ref{eq:MDrate}). Then, there exists some $\Lambda \in \calL$ such that for each $i$ the following holds: 
\begin{itemize}
\item if $\left((\rho_{n,i})_{n=1}^\infty, \left(\frac{1}{n}\log\frac{1}{\Lambda_{n,i}}\right)_{n=1}^\infty \right)$ satisfies moderate deviations regularity conditions then,
\begin{align}
&\liminf_{n \to \infty} \frac{1}{n(\rho_{n,i} - \frac{1}{n} \log \frac{1}{\Lambda_{n,i}})^2} \log \epsilon_{n,i} \geq -\frac{1}{2V_{\min}},
\end{align}
\item otherwise
\begin{align}
&\liminf_{n \to \infty} \epsilon_{n,i} > 0.
\end{align}
\end{itemize}
\end{theorem}

Here the tradeoffs are not apparent if $m_n$ growing exponentially in $\sqrt{n}$ since then $\frac{1}{n}\log \frac{1}{\Lambda_{n,i}} = o(\rho_{n,i})$ for all valid $\rho_{n,i}$. If it is growing any faster, however, we can observe degradation in the speed of convergence to the moderate deviations exponent. Thus, the moderate deviations setting interpolates the loss observed for fixed error asymptotic and error exponent regimes.

\section{Concluding Remarks}
\label{sec:conclusion}

%
Throughout this paper we have used the set $\calL_m$ and its asymptotic counterpart $\calL$ to capture the tradeoffs between different message classes in a UMP code. It may be useful to give an intuitive interpretation of the $\calL_m$ set. We interpret each element of $\calL_m$ as capturing a partitioning of `resources' (e.g., decoding space) between different classes. This is the main idea behind our converse bound of Theorem~\ref{thm:meta2}; there the common output distribution $Q_Y$ is used to tie the $m$ sub-codes together. The same parameterization appears in our achievability bounds of Corollary~\ref{cor:DTbnd}  and Theorem~\ref{thm:kappabeta}. This suggests that such resource `sharing' can be accomplished in a rather efficient way. Next, we may wonder if UMP codes parameterized by one element of $\calL_m$ are better or worse than codes parameterized by another element of $\calL_m$. To answer this question it is helpful to relate them to some operational quantity. This is discussed next.

\subsection{Operational Meaning of $\calL_m$}
Recall from Section~\ref{sec:problem} that one measure of ``goodness'' proposed for UMP codes is the expected rate (see Definition~\ref{def.expectedRate}). Suppose we fix $m$ error probability constraints $(\epsilon_i)_{i=1}^m$ and study the corresponding possible sizes of $m$ message classes. The finite block length bounds tell us that given the $(\epsilon_i)_{i=1}^m$ constraints there is a family of UMP codes parametrized by $\calL_m$. W nat wish to maximize the expected rate over this family of codes. Ignoring the third order terms in Theorem~\ref{thm:asmp} we obtain the following normal approximation for the size of each code at finite $n$ for a given $\blambda \in \calL_m$,
\begin{align}
\log M_i \approx nC - \sqrt{nV} Q^{-1} (\epsilon_i) - \log \frac{1}{\lambda_{i}}, \quad 1\leq i \leq m. 
\end{align}
Let us fix some prior probabilities $(\mu_1, \dots, \mu_m)$ on the $m$ message classes and consider maximizing the expected rate given $(\epsilon_i)_{i=1}^m$,
\begin{align}
 & \max_{\blambda \in \calL_m} R ({\bm \mu})=  \max_{\blambda \in \calL_m} \frac{1}{n} \sum_{i=1}^m \mu_i (\log M_i - \log \mu_i) \\
 & \approx  \max_{\blambda \in \calL_m} \frac{1}{n}\sum_{i=1}^m \mu_i ( nC - \sqrt{nV} Q^{-1} (\epsilon_i) - \log \frac{1}{\lambda_{n,i}} - \log\mu_i )\label{eqn:subs_logM}
\end{align}
The first two terms in \eqref{eqn:subs_logM} are constant since they do not involve $\bm{\lambda}$. Let $A = \frac{1}{n}\sum_{i=1}^m \mu_i (nC- \sqrt{n} Q^{-1}(\epsilon_i))$. Then, we have 
\begin{align}
  \max_{\blambda \in \calL_m} R ({\bm \mu}) &=  A+\max_{\blambda\in\calL }\frac{1}{n} \sum_{i=1}^m \mu_i \log \lambda_{i} -  \frac{1}{n}\sum_{i=1}^m \mu_i \log \mu_{i} \\
  &= A + \frac{1}{n}\sum_{i=1}^m \mu_i  \log\mu_i  - \frac{1}{n} \sum_{i=1}^m \mu_i \log \mu_{i} = A. \label{eqn:prop_betting}
\end{align}
Equation~\eqref{eqn:prop_betting} follows from the fact that the $\blambda$ that maximizes the expected rate over $\calL_m$ is given by proportional betting with $\lambda_i = \mu_i$ for all $1 \le  i\le m$~\cite[Theorem 6.1.2]{book:CT}. In other words, the UMP code that maximizes the expected rate given a prior message class distribution is one with $\lambda_i = \mu_i$. Of course, if we pick any other code we would suffer a loss of $D(\bm \mu || \blambda)$ in terms of expected rate. A more formal study of this connection is left to future work.

\subsection{Major Contributions and Future Work}

The main contribution of this paper is a collection of theorems which quantify tradeoffs involved in unequal message protection in asymptotic and non-asymptotic settings. We present extensions of well known finite block length bounds to UMP codes and demonstrate that both converse and achievability bounds admit similar tradeoffs which are captured by the probability simplex $\calL_m$. Although there is a gap between these bounds at finite block lengths (just as in the original bounds), they are shown to be tight in fixed error and moderate deviations asymptotic regimes. Our results also elucidate why tradeoffs inherent to unequal message protection were not observed in previous works on the subject. In each case this was due either to the asymptotic regime studied, the scaling of the number of classes with $n$, or both. In addition to exposing some fundamental tradeoffs of channel coding with unequal message protection this paper raises a number of follow up questions.

{\bf Channels with cost:} One interesting question not addressed in this paper is unequal message protection for channels with cost. Our $\kappa\beta$-bound extension in Theorem~\ref{thm:kappabeta} and converse bound in Corollary~\ref{cor:constant} could be applied to this problem when the cost constraint is the same for all $m$ classes. However, the most general formulation of channels with cost should involve different constraints for each class. Although the extension of the $\kappa\beta$-bound to such a setting would be  quite interesting, one would likely get more utility out of extending the DT bound with cost constraints~\cite[Theorem 24]{thesis:Polyanskiy} using similar approach to one used in Theorem~\ref{thm:DTbnd}. Likewise, a question arises as to how evaluate a meta-converse type bound since different cost constraints would have different `good' output distributions $Q_Y$. One possible approach is to evaluate the UMP meta-converse $m$ times, using the `best' $Q_Y$ for each class, and take the intersection over the regions obtained. 

In general, we can expect for UMP codes with cost constraints to behave in the following way. When the cost constraints are similar we will approach results derived in this paper where the loss is captured by the set $\calL_m$. In a case when the cost constraints are drastically different the codes will approach the no-loss  setting. Consider, for example, a two-class UMP code for an AWGN channel with power constraints $P_1$ and $P_2$. If $P_1 \approx P_2$ both sub-codes will reside on approximately the same sphere determined by the power constraint. The channel noise will thus push codewords from both sub-codes into the same decoding space. If $P_1 << P_2$ they will reside on power spheres that are very far apart making it so that the two sub-codes are very easy to distinguish at the channel output. 

{\bf Asymptotic theorems for mixed regimes:} To motivate this asymptotic setting let us consider {\em red alert} codes studied in~\cite{BNZ2009, mine:NSD2013}. A red alert code is a type of UMP code that has two classes. One class has a single extremely well protected ``red alert'' codeword. The other class has exponentially many normal codewords that have some reasonable amount of error protection. In the context of streaming communication with feedback the red alert codeword can be used to signal the decoder  a potentially erroneous decision, while normal codewords are used to achieve high communication rate~\cite{Kudryashov1979, mine:SD2010, mine:SDN2011}. Guided by this motivation we would like the asymptotics of such a code to behave in the following way.  For the red alert codeword we want the rate to be fixed (in this case at zero), and the probability of error to drop as fast as possible; this is reminiscent of the error exponent regime. For the normal codewords we can tolerate a small but non-zero error probability while we want the rate to approach capacity as fast as possible: this is exactly the setting for fixed-error asymptotics. 

In this work we follow the philosophy of previous asymptotic works in\cite{Csiszar1982, BNZ2009, WIK2011} and focus our attention on sequences of codes within one regime only. For example, Theorem~\ref{thm:asmp} assumes that all classes in a sequence of UMP codes have constant error probability. Likewise, in Theorem~\ref{thm:md} we assume that the rates of all the classes approach capacity at a rate consistent with the moderate deviations setting studied in~\cite{AltWag2010, preprint:AW2012,PolVer2010}. As the first study of tradeoffs for UMP codes this has the advantage of letting us compare our bounds to the homogeneous setting. The red alert example, however, brings up a rather subtle issue that is not present in the classical channel coding. It is entirely possible to have a sequence of UMP codes in which rates (resp. errors) of different classes approach capacity (resp. zero) at different speeds, or not at all. Moreover, in light of this example, these sequences of codes may have very interesting applications.  Studying the mixed setting is, thus, a natural next step.

{\bf Construction of practical UMP codes:}
Due to their connection to problems like streaming communication and joint source-channel coding, UMP codes may prove to be useful communication tools. Practical design of UMP codes poses a compelling question. As we have shown in Section~\ref{sec:dmc} in our discussion of the header construction simply taking existing codes and combining them first to encode the message class, and then encode the message, may not yield a good enough solution. Instead, a more intricate ``mixing'' of codewords is desired. Understanding how to construct such codes with practical construction schemes such as LDPC, Turbo, or Polar codes poses an interesting coding problem. Likewise, constructing decoding algorithms for such codes could prove to be a separate challenge. For example, the decoding complexity for UMP codes may scale with the number of classes, as in Theorem~\ref{thm:coset}. On the other hand, it may be possible to avoid such scaling through smart algebraic design.

Finally, other extensions of this problem may be of interest. A natural dual question to UMP codes would be source coding with unequal distortion criterion where some sources receive better distortion guarantees than other, an idea also proposed in~\cite{BNZ2009}. The connection between UMP codes and joint source-channel coding is the most natural direction of study. The idea of using UMP codes for joint source-channel coding will be explored in further detail in subsequent work.
\appendices


\section{Proofs for Finite Block Length Bounds}
\label{appnd:finite}
\begin{proof}[Proof of Theorem~\ref{thm:DTbnd}]
We first describe the operation of the decoder for a given UMP codebook $\calC$. Then, we outline a codebook construction based on a {\em sequential random coding} technique. The error analysis will be done simultaneously with the codebook construction.

{\em Decoding:} We will use a sequential threshold decoder. Specifically, the decoder computes $\imath_{X_i;Y_i}(c_{i,w}, y)$ for received channel output $y$ where $i$ varies from $1$ to $m$, and $w$ varies from $1$ to $M_i$. The decoder outputs the first codeword for which $\imath_{X_i;Y_i}(c_{i,w},y) > \log \tau_i(x)$. Formally, the decoder is defined as
\begin{align}
g(y) = \argmin_{i,w} \left\{c_{i,w} : \imath_{X_i;Y_i}(c_{i,w},y) > \log \tau_i(x) \right\}, \label{eq:DTdec}
\end{align}
where $i \in \{1, \dots, m\}$ and $w \in \{1,\dots, M_i\}$. 

{\em Codebook Construction:} We construct a codebook sequentially starting with codewords in class $1$, then class $2$, all the way to class $m$. To select $c_{1,1}$ choose $x$ at random with distribution $P_{X_1}$. Then
\begin{align}
\bbE \left[\epsilon_{1,1} (x) \right]= \bbP \left[\imath_{X_1; Y_1}(X_1,Y_1) \leq \log \tau_1(X_1)\right].
\end{align}
There must exist at least one $x$ such that $\epsilon_{1,1}(x) \leq \bbP \left[\imath_{X_1;Y_1}(X_1,Y_1) \leq \log \tau_1(X_1)\right]$. Call this $c_{1,1}$ and go on to select $c_{1,2}$ all the way to $c_{1, M_1}$. 

Suppose the sub-codebooks for the first $i-1$ classes, $\calC_1, \dots, \calC_{i-1}$, have been selected, as well as $l$ codewords in $\calC_i$ for some $1\leq i\leq m$ and $0 \leq l \leq  M_i -1$. We show that we can add a codeword to $\calC_i$ without violating~(\ref{eq:DTbnd}). Denote
\begin{align}
D_j &= \bigcup_{w = 1}^{M_j} \left\{y : \imath_{X_i;Y_i}(c_{j,w},y) >  \log\tau_j(c_{j,w}) \right\}, 
\end{align}
for $1\leq j \leq i-1$ and
\begin{align}
D_i &= \bigcup_{w = 1}^{l} \left\{y : \imath_{X_i;Y_i}(c_{i,w},y) >  \log \tau_i (c_{i,w})\right\}.
\end{align}
Select $c_{i, l+1}$ by choosing $x$ at random with distribution $P_{X_i}$. Then
\begin{align}
&\bbE\left[\epsilon_{i,l+1} (c_{1,1}, \dots, c_{i, l}, x)\right] \notag\\
& = \bbP\left[\bigcup_{j=1}^i D_j \cup \left\{\imath_{X_i;Y_i}(X_i,Y_i) \leq \log \tau_i(X_i) \right \} \right] \\
& \leq \bbP \left[ \imath_{X_i;Y_i}(X_i,Y_i) \leq \log\tau_i (X_i)\right] + \sum_{j=1}^i \bbP \left[D_j \right]\label{eq:DTunion1}\\
&\leq  \bbP \left[\imath_{X_i;Y_i}(X_i,Y_i) \leq \log\tau_i (X_i)\right] + (M_i -1) \sup_x \bbP\left[\imath_{X_i;Y_i}(x; Y_j) >  \log \tau_j(x) \right] \notag\\
&\quad +\sum_{j=1}^{i-1} M_j\sup_x \bbP \left[ \imath_{X_i; Y_i}(x,Y_j) \geq \log\tau_j(x) \right] \label{eq:DTunion2}
\end{align}
where~(\ref{eq:DTunion1}) and~(\ref{eq:DTunion2}) both follow by union bound. There must be at least one $x$ such that $\epsilon_{i,l+1} (c_{1,1}, \dots, c_{i, l}, x)$ is less than~(\ref{eq:DTunion2}): call this $c_{i, l+1}$. Finally, the encoder maps $w$th message in $\mathcal{M}_i$ to $c_{i,w}$, and the decoder maps $c_{i,w}$ to $w$th message in $\mathcal{M}_i$ which gives the result.
\end{proof}

\begin{proof}[Proof of Theorem~\ref{thm:kappabeta}]
We first describe the decoder for a given UMP codebook  $\mathcal{C}$. We then use induction on the number of message classes to show that a codebook satisfying~(\ref{eq:KappaBetaExt}) and~(\ref{eq:KappaBetaExt2}) can be constructed.
 
{\bf Decoding: }  Given an output $y \in \mathsf{B}$ the decoder sequentially tests whether $c_{i,w}$ was sent with $i$ running from $1$ to $m$, and $w$ running from $1$ to $M_i$. The test for $c_{i,w}$ is performed as a binary hypothesis test discriminating $W_{c_{i,w}}$ (hypothesis $\mathcal{H}_1$) against ``average noise'' $Q_Y$ (hypothesis $\mathcal{H}_0$). Given class $i$ we would like to select each such test as an optimal one with the constraint $P(\mbox{decide } \mathcal{H}_1 | \mathcal{H}_1) \geq 1-\epsilon_i+\tau$. To do this we define $m$ collections of random variables $Z_i(x)$, $x\in \mathsf{F}$ all conditionally independent given $Y$ and with $P_{Z_i(x)|Y}$ chosen so that it achieves $\beta_{1-\epsilon_i+\tau}(W_x, Q_Y)$. In other words,
\begin{align}
P\left[ Z_i(x)  = 1 | X=x\right] \geq 1-\epsilon_i + \tau,\\
Q\left[ Z_i(x)  = 1\right] = \beta_{1- \epsilon_i + \tau}(W_x, Q_Y),
\end{align}
which we can do by the Newman-Pearson Lemma. 

The decoder applies independent random transformations $P_{Z_1}(c_{1,1}), \dots, P_{Z_1}(c_{1,M_1})$ to output $Y$, then 

$P_{Z_2}(c_{2,1}), \dots, P_{Z_2}(c_{2,M_2})$, and so on for all $m$ classes. It outputs the fist index $(i, w)$ for which $Z_i(c_{i,w}) = 1$.

We proceed to prove the rest of the theorem via induction. 

{\bf Codebook Construction:}
To show the claim for $m=1$ we have that for an UMP code with one message class
\begin{align}
M_1 \geq \frac{\kappa_\tau (\mathsf{F}, Q_Y)}{\sup_{x\in \mathsf{F}} \beta_{1-\epsilon_1+\tau}(W_x, Q_Y)}
\end{align}
by appealing to~\cite[Theorem 27]{thesis:Polyanskiy}. It follows that there must exist and $\left(M_1, \epsilon_1\right)$-UMP code satisfying~(\ref{eq:KappaBetaExt2}) with $0\leq \lambda_1 \leq 1$.

Let us assume the theorem statement is true for $m-1$ and fix arbitrary $\bm \lambda \in \mathcal{L}_m$. By inductive hypothesis we can construct  $\left( (M_i)_{i=1}^{m-1}, ( \epsilon_i)_{i=1}^{m-1} \right) $-UMP code such that
\begin{align}
M_i =  \left \lfloor \frac{ \lambda_i \kappa_\tau (\mathsf{F}, Q_Y)}{\sup_{x\in \mathsf{F}} \beta_{1-\epsilon_i+\tau}(x, Q_Y)} \right \rfloor.
\end{align}

If $\calC_1, \ldots, \calC_{m-1}$ are the sub-codebooks associated with this code, we can construct $\mathcal{C}_m$ by rehashing the greedy approach of~\cite[Theorem 27]{thesis:Polyanskiy}. Suppose $j$ codewords have already been selected for $\calC_m$ (where $j$ could be zero). Define
\begin{align}
U_i &= \max \{Z_i(c_{i,1}), \dots, Z_i (c_{i, M_i})\} \quad \mbox{for } 1\leq i \leq m-1, \\
V_j &= \max \{0,Z_m(c_{m,1}), \dots, Z_k(c_{m,j})\}.
\end{align}
We choose the $j+1$-st codeword by selecting an arbitrary $x \in \mathsf{F}$ which satisfies
\begin{align}
\bbP\left[Z_m(x) =1, U_1 = \dots = U_{m-1} = V_j = 0 |X=x \right] \geq 1-\epsilon_m.
\end{align} 
Once no such $x \in \mathsf{F}$ can be found, we stop. 

{\bf Relating Error to Codebook Size:} Suppose the process stops after $M_m$ steps and let 
\begin{align}
Z =\max (U_1, \dots, U_m)
\end{align}
where $U_m = V_{M_m}$. 
This implies that for every $x \in \mathsf{F}$ we have
\begin{align}
\bbP\left[Z_m(x) =1, Z=0 |X=x \right] < 1-\epsilon_m
\end{align}
Then by definition of $Z_m(x)$ it follows
\begin{align}
&1- \epsilon_m + \tau \leq \bbP[Z_m(x) = 1| X=x]\\
&=  \bbP[Z_m(x) = 1, Z = 0| X=x]  + \bbP[Z_m(x) = 1, Z= 1| X=x]  \\
&\leq  \bbP[Z_m(x) = 1, Z= 0| X=x] + \bbP[Z = 1| X=x] \\
&< 1-\epsilon_m+ \bbP[Z= 1| X=x].
\end{align}
So, for every $x \in \mathsf{F}$
\begin{align}
\bbP[Z = 1| X=x] \geq \tau.
\end{align}
This is exactly the composite hypothesis test defined in~(\ref{eq:kappa}) and 
\begin{align}
Q[Z  = 1] \geq \kappa_\tau(\mathsf{F}, Q_Y).
\end{align}
Finally, we can bound
\begin{align}
&Q[Z = 1] = Q \left[ \bigcup_{i=1}^{m} \{U_i = 1\} \right]\\
&\leq \sum_{i=1}^{m} Q \left[U_i = 1 \right]\\
&\leq \sum_{i=1}^m Q \left[ \bigcup_{w=1}^{M_i} \{Z_i(c_{i,w}) = 1 \} \right]\\
&\leq \sum_{i=1}^m \sum_{w=1}^{M_i} Q\left[  \{Z_i(c_{i,w}) = 1 \} \right] \\
&= \sum_{i=1}^m \sum_{w=1}^{M_i} \beta_{1-\epsilon_i+\tau}(W_{c_{i,w}}, Q_Y)\\
&\leq \sum_{i=1}^{m-1}  M_i\sup_{x \in \mathsf{F}} \beta_{1-\epsilon_i+\tau}(W_x, Q_Y) +M_m\sup_{x \in \mathsf{F}} \beta_{1-\epsilon_m +\tau}(W_x, Q_Y)\\
&\leq \sum_{i=1}^{m-1} \lambda_i \kappa_{\tau}(\mathsf{F},Q_Y) + M_m \sup_{x \in \mathsf{F}} \beta_{1-\epsilon_m+\tau}(W_x, Q_Y).
\end{align}
Thus, we conclude that
\begin{align}
 M_m \sup_{x \in \mathsf{F}} \beta_{1-\epsilon_m+\tau}(W_x, Q_Y) \geq \lambda_m \kappa_{\tau}(\mathsf{F},Q_Y)
\end{align}
and that there exists an UMP code with $m$ classes of codewords satisfying~(\ref{eq:KappaBetaExt}) and~(\ref{eq:KappaBetaExt2}).
\end{proof}

\begin{proof}[Proof of Theorem~\ref{thm:DTavg}]
To show~(\ref{eq:experrDT}) we generate the codewords in each sub-code $\calC_i$ as independent random variables with common distribution $P_{X_i}$ and use the decoding rule defined in~(\ref{eq:DTdec}).
Let $E(\bm{\mu})$ be the random variable denoting the expected error and $E_i$ the random variable denoting the average error for class $i$ across the ensemble of all codebooks. Then
\begin{align}
\bbE[E(\bm{\mu})] = \bbE \left[\sum_{i=1}^m \mu_i E_i \right] =  \sum_{i=1}^m \mu_i \bbE \left[ E_i \right]. \label{eq:eeDT}
\end{align}
The average error for each class can be bound as
\begin{align}
\bbE [E_i] &\leq \bbP \left[\imath_{X_i;Y_i}(X_i;Y_i) \leq \log\tau_i (X_i) \right] +(M_i -1) \bbP\left[\imath_{X_i; Y_i}(\barX_i; Y_i) >  \log \tau_i (x) \right] \notag\\
&\quad + \sum_{j = 1}^{i-1} M_j \bbP\left[\imath_{X_j; Y_j}(X_j; Y_i) >  \log \tau_j(x) \right] 
\end{align}
where $P_{\barX_i, Y_i} (x,y) = P_{X_i}(x) P_{X_i}W(y)$ and $P_{X_j, Y_i} (x,y) = P_{X_j}(x) P_{X_i}W(y)$ as in~\cite[Theorem 18]{PPV2010}. Following reasoning similar to Corollary~\ref{cor:DTbnd} we obtain
\begin{align}
 \bbE [E_i]  \leq  \bbE \left[ \exp\left\{ -\left[\imath_{X_i;Y_i}(X_i;Y_i)-\log \frac{M_i}{\lambda_i}\right]^+\right\} \right] \label{eq:eeDT2}.
\end{align} 
Combining~(\ref{eq:eeDT}) with~(\ref{eq:eeDT2}) gives the result and applying Shannon's argument we conclude that there exists a code satisfying~(\ref{eq:experrDT}).
\end{proof}

\begin{proof}[Proof of Theorem~\ref{thm:RCU}]
To show~(\ref{eq:experrRCU}) we generate the codewords in each subcode $\calC_i$ as independent random variables with common distribution $P_{X_i}$. Denote the codewords in class $i$ by $X_{i,1}, \dots, X_{i, M_i}$. Our decoding rule is to pick the codeword with largest biased information density,
\begin{align}
\argmax_{i, w} \left\{\log \tau_i + \imath_{X_i;Y_i}(X_{i,w}; Y_i) \right\}
\end{align}
where $i \in \{1, 2, \dots, m\}$ and $w \in \{1, 2, \dots, M_i\}$.
Let $E(\bm{\mu})$ be the random variable denoting the expected error and $E_i$ the random variable denoting the average error for class $i$ across the ensamble of all codebooks. Then
\begin{align}
\bbE[E(\bm{\mu})] = \bbE \left[\sum_{i=1}^m \mu_i E_i \right] =  \sum_{i=1}^m \mu_i \bbE \left[ E_i \right]. \label{eq:experr}
\end{align}

To bound the average error for each class suppose the first codeword from class $i$ was sent. An error occurs only if the biased information density for some other codeword. By symmetry we obtain
\begin{align}
\bbE[E_{i}] &\leq \bbP \left[  \bigcup_{j=1}^m \bigcup_{\begin{array}{c} w=1, \\(j,w) \neq (i,1) \end{array}}^{M_j} \big\{\log \tau_j +\imath_{X_j; Y_j}(\bar{X}_{j,w}, Y_i)\geq \log \tau_i+ \imath_{X_i;Y_i}(X_{i,1}, Y_{i,1}) \big\} \right]\\
&= \bbE \left[ \bbP \left[ \left. \bigcup_{j=1}^m \bigcup_{\begin{array}{c} w=1, \\(j,w) \neq (i,1) \end{array}}^{M_j} \left\{\imath_{X_j;Y_j}(\barX_{j,m};Y_i)\geq (\log \tau_i - \log \tau_j )+ \imath_{X_i;Y_i}(X_{i,1}; Y_i) \right\}\right| X_{i,1}, Y_{i,1} \right]\right] \\
&\leq \bbE \left[ \min \left\{ 1, \sum_{j=1}^m \left(M_j-\Indic{1 \neq j} \right) \bbP \left[ \left.\left\{\imath_{X_j;Y_j}(\bar{X}_j, Y_i)\geq (\log \tau_i - \log \tau_j )+ \imath_{X_i;Y_i}(X_i, Y_i) \right\}\right| X_i, Y_i \right] \right\}\right]. \label{eq:expclass}
\end{align}
Combining~(\ref{eq:experr}) with~(\ref{eq:expclass}) we get
\begin{align}
\bbE[\epsilon(\bm \mu)]  \leq \sum_{i=1}^m \mu_i  \bbE \left[ \min \left\{ 1, \sum_{j=1}^m \left(M_j-1_{i\neq j}\right) \bbP \left[ \left.\left\{\imath_{X_j;Y_j}(\bar{X}_j, Y_i)\geq (\log \tau_i - \log \tau_j )+ \imath_{X_i;Y_i}(X_i, Y_i) \right\}\right| X_i, Y_i \right] \right\}\right]
\end{align}
and conclude that there exists at least one codebook with $\epsilon(\bm \mu)$ satisfying~(\ref{eq:experrRCU}).
\end{proof}

\section{Utility Theorems}
We use the theorems in this section to prove our asymptotic results. All theorems have the following common set up.

Let $Z_j$, $j=1, \dots, n$ be independent random variables with
\begin{align}
\mu_j = \bbE[Z_j], \quad \sigma_j^2 = Var[Z_j], \quad \mbox{and } t_j = \bbE[|Z_j - \mu_j|^3].
\end{align} 
Denote $V = \sum_1^n \sigma_j^2$ and $T = \sum_1^n t_j$.

\begin{theorem}[Berry-Esseen]
\label{thm:BE}
\begin{align}
\left| \bbP \left[ \frac{\sum_{j=1}^n (Z_j - \mu_j) }{ \sqrt{V}} \leq \lambda\right] - Q(-\lambda) \right| \leq \frac{T}{V^{3/2}}
\end{align}
\end{theorem}

The following theorem is a refined version of the Berry-Esseen theorem.
\begin{theorem}[Rozovsky]
\label{thm:Roz}
Assume $Z_j$ have finite third moments, that is $t_j < \infty$. Then there exist universal constants $A_1>0$ and $A_2>0$ such that whenever $x\geq 1$ we have
\begin{align}
 \bbP \left[ \frac{\sum_{j=1}^n (Z_j - \mu_j) }{ \sqrt{V}} > x \right]   \geq Q(x)\exp\left\{ -\frac{-AT}{V^{3/2}}x^3\right\}\left(1-\frac{A_2T}{V^{3/2}} \right).
\end{align}
\end{theorem}
See~\cite{Rozovsky2002} for proof.

\begin{theorem}[Polyanksiy-Poor-Verd\'{u}]
\label{thm:PPV}
Assume $V>0$ and $T\leq \infty$. Then for any $A$
\begin{align}
&\bbE \left[ \exp \left\{ - \sum_{j=1}^n Z_j \right\} \Indic{\sum_{j=1}^n Z_j > A}\right] \leq 2 \left( \frac{\log2}{\sqrt{2\pi}} + \frac{12T}{V}\right) \frac{1}{\sqrt{V}} \exp\{-A\}
\end{align} 
\end{theorem}
See~\cite[Lemma 20]{thesis:Polyanskiy} for proof.

\section{Proof of Theorem~\ref{thm:asmp} - Fixed Error Asymptotics}

\label{appnd:asymptotic}
\begin{proof}[Achievability Proof of Theorem~\ref{thm:asmp}]
Fix some $\Lambda \in \cal{L}$. For each sufficiently large block length $n$ we will apply Theorem~\ref{thm:DTbnd} with $\rvA = \calA^n$, $m=m_n$, and $P_{X_i^n} = P^{\epsilon_i}_{X^n}$. $P^{\epsilon_i}_{X^n} \in \Pi$ is the distribution that achieves $V_{\min}$ if $\epsilon < 1/2$ and it is the distribution that achieves $V_{\max}$ otherwise. Observe that
\begin{align}
\imath_{X_i^n;Y_i^n}(X_i^n, Y_i^n) = \sum_{j=1}^n \log \frac{W(Y_{i,j},X_{i,j})}{P^{\epsilon_i}_{X}W(Y_{i,j})} = \sum_{j=1}^n Z_{i,j}, \quad \forall i.
\end{align}
Then for all $i \in \{1, \dots, m_n\}$ and $j \in \{1, \dots, n\}$,
\begin{align}
\bbE[Z_{i,j}] &= I(P^{\epsilon_i}_{X},W), \\
\bbV ar(Z_{i,j}) &= V(P^{\epsilon_i}_{X},W), \label{eq:VisU}\\
\mbox{and } \kappa &= \sum_{x,y} P^{\epsilon_i}_X(x) W(y|x) \left| \log\frac{W(y|x)}{P^{\epsilon_i}_XW(y)} - I(P^{\epsilon_i}_X,W) \right|^3 \leq \infty.
\end{align}
where~(\ref{eq:VisU}) follows by~\cite[Lemma 46]{thesis:Polyanskiy} since $P^{\epsilon_i}_{X}$ is the capacity achieving distribution.

For $n$ sufficiently large and $i \leq m_n$ define a sequence of constants $\tilM_{n,i}$ such that
\begin{align}
\log \tilM_{n,i} = nC - Q^{-1}(\epsilon_i - 3\delta_n)\sqrt{nV(P^{\epsilon_i}_X, W)} \geq 0 \label{eq:tilM}
\end{align}
where 
\begin{align}
\delta_n=  2 \left( \frac{ \log 2}{\sqrt{\pi V(P^{\epsilon_i}_X, W)}} + 2B\right) \frac{1}{\sqrt{n}}, \quad \mbox{ and } \quad B = \frac{2^{3/2} 6 \kappa}{V(P^{\epsilon_i}_X,W)^{3/2}}.
\end{align}
Note that $\delta_n$ depends on the channel, but not $i$, and goes to zero as $\frac{1}{\sqrt{n}}$.

Finally, select the decoding thresholds
\begin{align}
\tau_{n,i}(x^n) = \left \{ \begin{array}{cc}
\tilM_{n,i},  & Var[\imath_{X^n;Y^n}(X^n; Y^n) |X^n = x^n] \geq \frac{nV(P^{\epsilon_i}_X,W)}{2}, \\
+\infty, & \mbox{otherwise}. \end{array}\right. \label{eq:threshold}
\end{align}

Theorem~\ref{thm:DTbnd} guarantees an existence of $\left( (M_{n,i})_{i=1}^{m_n},  (e_{n,i})_{i=1}^{m_n} \right)$-UMP code (maximum probability of error) with
\begin{align}
e_{n,i} \leq \bbP \left[\imath_{X_i^n;Y_i^n}(X_i^n;Y_i^n) \leq \log\tau_{n,i}(X^n) \right] + \sum_{j = 1}^{i} M_{n,j} \sup_{x^n\in \calA^n} \bbP\left[\imath_{X_j^n;Y_j^n}(x^n; Y_i^n) >  \log \tau_{n,j}(x^n) \right].  
\end{align}
We will show that $e_{n,i} \leq \epsilon_i$ for $M_{n,i} \geq \tilM_{n,i}$, for all $i$ and $n$ sufficiently large.

The first term is upper-bounded as follows:
\begin{align}
\bbP \left[\imath_{X_i^n;Y_i^n}(X_i^n;Y_i^n) \leq \log\tau_{n,i}(X_i^n) \right] &\leq \bbP \left[\imath_{X_i^n;Y_i^n}(X_i^n; Y_i^n) \leq \log  \tilM_{n,i} \right] + \bbP[\tau_{n,i} =\infty] \label{eq:fterm1}\\
&\leq \bbP \left[\frac{\imath_{X_i^n;Y_i^n}(X_i^n, Y_i^n) -nC}{\sqrt{nV(P^{\epsilon_i}_X,W)}}\leq -Q^{-1} (\epsilon_i - 3\delta_n) \right] + \bbP[\tau_{n,i} =\infty] \label{eq:fterm2}\\
&\leq \epsilon_i -3\delta_n + \frac{B}{\sqrt{n}} +   \bbP[\tau_{n,i} =\infty] \label{eq:fterm3}\\
&\leq \epsilon_i -2\delta_n + \left(\frac{B}{\sqrt{n}} +    \exp\{-O(n)\} - \delta_n \right)\label{eq:fterm4}\\
&\leq \epsilon_i -2\delta_n \label{eq:fterm5}
\end{align}
where~(\ref{eq:fterm3}) follows by appealing to Theorem~\ref{thm:BE},~(\ref{eq:fterm4}) follows by Chernoff bound applied to a sum of bounded \iid random variables, and~(\ref{eq:fterm5}) follows for $n$ sufficiently large (where ``$n$ sufficiently large'' depends on channel only). 

To bound the second term we first bound each term in the sum as follows: 
\begin{align}
\sup_{x^n} & \bbP\left[\imath_{X_j^n; Y_j^n}(x^n; Y_i^n) >  \log \tau_{n,j}(x^n) \right] \leq  \sup_{\{x^n: \tau_{n,j} < \infty\}} \bbP\left[\imath_{X_j^n; Y_j^n}(x^n; Y_i^n) >  \log\tilM_{n,j} \right] \label{eq:sterm1}\\
&=  \sup_{\{x^n: \tau_{n,j} < \infty\}} \bbE\left[ \Indic{ \imath_{X_j^n; Y_j^n}(x^n; Y_i^n) >  \log\tilM_{n,j}}\right] \label{eq:sterm2}\\
&=  \sup_{\{x^n: \tau_{n,j} < \infty\}}  \sum_{y^n \in \calB^n} W(y^n|x^n)  \frac{ P^{\epsilon_i}_XW(y^n)}{W(y^n|x^n)}\left[ \Indic{\imath_{X_j^n; Y_j^n}(x^n; Y_i^n) >  \log\tilM_{n,j}}\right] \label{eq:sterm3}\\
&=  \sup_{\{x^n: \tau_{n,j} < \infty\}}  \sum_{y^n \in \calB^n} W(y^n|x^n)  \frac{ P^{\epsilon_j}_XW(y^n)}{W(y^n|x^n)}\left[ \Indic{\imath_{X_j^n; Y_j^n}(x^n; Y_j^n) >  \log\tilM_{n,j}}\right] \label{eq:sterm3b}\\
&=  \sup_{\{x^n: \tau_{n,j} < \infty\}}  \bbE \left[\exp\{-\imath_{X_j^n; Y_j^n}(x^n; Y_j^n)\}  \Indic{ \imath_{X_j^n; Y_j^n}(x^n; Y_j^n) >  \log\tilM_{n,j}}\right] \label{eq:sterm4}\\
&\leq  \sup_{\{x^n: \tau_{n,j} < \infty\}}  \frac{1}{\tilM_{n,j}}2\left(\frac{\log2}{\sqrt{2\pi}} +\frac{12T}{\bbV ar[\imath_{X_j^n; Y_j^n}(X_j^n; Y_j^n) |X_j^n = x^n]} \right) \frac{1}{\bbV ar[\imath_{X_j^n; Y_j^n}(X_j^n;Y_j^n) |X_j^n = x^n]} \label{eq:sterm5}\\
&\leq \frac{1}{\tilM_{n,j}}2\left(\frac{\log2}{\sqrt{\pi V(P^{\epsilon_j}_X,W)}} +2B \right)\frac{1}{\sqrt{n}} =  \frac{1}{\tilM_{n,j}} \delta_n \label{eq:sterm6}
\end{align}
where~(\ref{eq:sterm2}) follows by rewriting a probability as an expectation of an indicator function,~(\ref{eq:sterm3}) follows by a change of measure argument,~(\ref{eq:sterm3b}) follows since the capacity achieving output distribution is unique,~(\ref{eq:sterm5}) follows by invoking Theorem~\ref{thm:PPV}, and~(\ref{eq:sterm6}) follows from~(\ref{eq:threshold}).

Now taking  $M_{n,i} = \lceil \Lambda_{n,i} \tilM_{n,i}\rceil$ for all $i$ we obtain
\begin{align}
\sum_{j = 1}^{i} M_{n,j} \sup_{x^n} \bbP\left[\imath_{X_j^n; Y_j^n}(x^n; Y_i^n) >  \log \tau_{n,j}(x^n) \right] \leq \sum_{j = 1}^{i}  \frac{M_{n,j}}{\tilM_{n,j}} \delta_n  < 2\delta_n .  
\end{align}
Thus, for $n$ sufficiently large, and $\log M_{n,i} \geq \log ( \Lambda_{n,i} \tilM_{n,i})$ we have $e_{n,i} \leq \epsilon_i$ for all $i$. The result follows by taking a Taylor expansion of~(\ref{eq:tilM}).
\end{proof}

\begin{proof}[Converse Proof of Theorem~\ref{thm:asmp}]
We start from the converse bound in Theorem~\ref{thm:meta2} with the particularizations $\rvA=\calA^n$ and  $m=m_n$. There it is shown that for any vector $\blambda=(\lambda_1,\ldots,\lambda_m)\in\calL_m$ and any  output distribution $Q_{Y^n}$ we have 
\begin{equation}
\log M_{n,i} \le \max_{P_{X^n} \in\calP(\calA^n)} \log\lambda_i - \log \beta_{1-\epsilon_i}(P_{X^n Y^n}^i , P_{X^n} \times Q_{Y^n})
\end{equation}
First, using Lemma 2 in \cite{TomTan2013}, we can further upper bound the above by 
\begin{equation}
\log M_{n,i} \le \max_{P_{X^n} \in\calP(\calA^n)} \log\lambda_i  + D_s^{\epsilon_i+\delta}(P_{X^n Y^n}^i  \, \|\,  P_{X^n}^i \times Q_{Y^n}) + \log\frac{1}{\delta} \label{eqn:bd_Mni}
\end{equation}
where the information spectrum divergence is defined in~(\ref{eq:InfSpcDiv}).

In fact this is the relaxation to the Verd\'u-Han converse lemma~\cite[Lem.~4]{VerHan1994}. By using Lemma 2 in \cite{TomTan2013}, we can evaluate  \eqref{eqn:bd_Mni} at for a particular input symbol independent of the input distribution (or code), i.e.\
\begin{equation}
\log M_{n,i} \le \max_{x^n \in\calA^n} \log\lambda_{n,i}  + D_s^{\epsilon_i+\delta}(W^n(\cdot|x^n)  \, \|\, Q_{Y^n}) + \log\frac{1}{\delta} \label{eqn:bd_Mni2}
\end{equation}
We will pick $\delta= n^{-1/2}$ and thus the final term is $\frac{1}{2}\log n$. The output distribution will be chosen to be \cite[Eq.~(6)]{TomTan2013}
\begin{equation}
Q_{Y^n} (y^n) =\frac{1}{2}\sum_{\bk\in\calK} \frac{\exp(-\gamma\|\bk\|_2^2)}{F}Q_{\bk}^n(y^n) + \frac{1}{2}\sum_{P\in\calP_n(\calA)} \frac{1}{|\calP_n(\calA)|} (PW)^n (y^n) \label{eq:TTout}
\end{equation}
where $F$ is  a normalization constant that ensures that $\sum_{y^n} Q_{Y^n}(y^n)=1$ and 
\begin{equation}
Q_{\bk}(y) := Q^*(y) + \frac{k_y}{\sqrt{n\zeta}},\qquad \calK:=\left\{ \bk\in\bbZ^{|\calY|}: \sum_y k_y=0, k_y \ge -Q^*(y)\sqrt{n\zeta}\right\}.
\end{equation}
As explained in \cite{TomTan2013}, this construction results an $n^{-1/2}$-net of distributions $\{Q_{\bk}\}_{\bk\in\calK}$ in the output  simplex. These output distributions serve to approximate those that are induced by an input type that is close to the capacity-achieving input distribution. We can then go through the same continuity arguments in Lemma 7 and Proposition 8 of \cite{TomTan2013} to conclude that  with this choice of output distributions, 
\begin{equation}
D_s^{e_i+\delta}(W^n(\cdot|x^n)  \, \|\, Q_{Y^n})\le n C - \sqrt{nV}\rmQ^{-1} (\epsilon_i)+O(1).
\end{equation}
for all $x^n\in \calA^n$. Putting all the pieces together, we have shown that 
\begin{equation}
\log M_{n,i} \le nC - \sqrt{nV}\rmQ^{-1} (\epsilon_i)+ \frac{1}{2}\log n - \log \frac{1}{\Lambda_{n,i}} +  O(1).
\end{equation}
To show the assertion for singular symmetric channels we pick output distribution as in~\cite{preprint:AltWag2013} and repeat the argument starting with~(\ref{eq:TTout}).
\end{proof}

\section{Proof of Theorem~\ref{thm:md} - Moderate Deviations Asymptotics}

\begin{proof}[Achievability Proof of Theorem~\ref{thm:md}]
Let $\Lambda \in \calL$ and a collections of sequences $\left( (\rho_{n,i})_{n=1}^\infty \right)_{i=1}^\infty$ be as required. Define
\begin{align}
M_{n,i} = \lfloor 2^{nC - n\rho_{n,i} } \rfloor
\end{align}
and
\begin{align}
\trho_{n,i} = \rho_{n,i} - \frac{1}{n}\log \frac{1}{\Lambda_{n,i}}.
\end{align}
Let $P^\ast_X$ be the capacity-achieving distribution which also achieves $V_{\min}$. Then, by~(\ref{eq:DTforMD}) there exists a sequence of $\Mepsilon$-UMP codes such that 
\begin{align}
\epsilon_{n,i}& \leq \bbP \left[\imath_{X^n_i; Y^n_i}(X^n_i;Y^n_i) \leq \log \tau_i(X^n) \right] + \frac{M_{n,i}}{\Lambda_{n,i}} \sup_{x^n} \bbP\left[\imath_{X^n_i;Y^n_i}(x^n; Y^n_i) >  \log \tau_i(x^n) )\right]\\
&= \bbE \left[\Indic{\imath_{X^n_i; Y^n_i}(X^n_i;Y^n_i) \leq \log \tau_i(X^n)}\right] + \frac{M_{n,i}}{\Lambda_{n,i}} \sup_{x^n} \bbE\left[ \Indic{\imath_{X^n_i;Y^n_i}(x^n; Y^n_i) >  \log \tau_i(x^n) )}\right].
\end{align}
Next, fix arbitrary $\gamma < 1$ and set $\log \tau_i(x^n) = nC-\gamma n (\rho_{n,i} -\frac{1}{n}\log \frac{1}{\Lambda_{n,i}})$ for all $x^n \in \calA^n$. Observe that it follows that, 
\begin{align}
\log \frac{M_{n,i}}{\Lambda_{n,i}} = \log M_{n,i} + \log \frac{1}{\Lambda_{n,i}}= nC - n\rho_{n,i}+ \log \frac{1}{\Lambda_{n,i}} = nC - n \trho_{n,i}.
\end{align}
For a fixed $x^n$ we get via a simple change of measure argument,
\begin{align}
 &\bbE  \left[ \frac{M_{n,i}}{\Lambda_{n,i}}\Indic{\imath_{X^n_i;Y^n_i}(x^n; Y^n_i) >  \log \tau_i(x^n) }\right]
 = \sum_{y^n \in \calB^n} \left[ \frac{M_{n,i}}{\Lambda_{n,i}} \Indic{\imath_{X^n_i;Y^n_i}(x^n; y^n_i) >  nC-\gamma n \trho_{n,i} }\right] P_{Y^n}(y^n)\\
&= \sum_{y^n \in \calB^n} \left[ \frac{M_{n,i}}{\Lambda_{n,i}} \left(\frac{P_{Y^n|X^n = x^n}(y^n)}{P_{Y^n}(y^n)}\right)^{-1}\Indic{\imath_{X^n_i;Y^n_i}(x^n; y^n_i) >  nC-\gamma n\trho_{n,i}}\right] P_{Y^n|X^n = x^n}(y^n)\\
&= \bbE \left[\exp \left\{-\left[\imath_{X^n_i; Y^n_i}(x^n, Y^n_i) - \log \frac{M_{n,i}}{\Lambda_{n,i}}\right]\right\} \Indic{\imath_{X^n_i;Y^n_i}(x^n; Y^n_i) >  nC-\gamma n\trho_{n,i} }\right] \\
& \leq \exp\{- (1-\gamma)n\trho_{n,i} \}.
\end{align}
And thus we get that
\begin{align}
\epsilon_{n,i} &\leq \bbP \left[\imath_{X^n_i; Y^n_i}(X^n_i;Y^n_i) \leq nC-\gamma n\trho_{n,i} \right] +\exp\{- (1-\gamma)n \trho_{n,i} \}\\
&\leq 2\max \left\{ \bbP \left[\imath_{X^n_i; Y^n_i}(X^n_i;Y^n_i) \leq nC-\gamma n\trho_{n,i} \right], \exp\{- (1-\gamma)n \trho_{n,i} \} \right\}
\end{align}
The result follows since by~\cite[Theorem 3.7.1]{book:DZ} and assumptions on $\trho_{n,i}$
\begin{align}
 \limsup_{n \to \infty} \frac{1}{n\trho_{n,i}^2} \log \bbP \left[\imath_{X^n_i; Y^n_i}(X^n_i;Y^n_i) \leq nC-\gamma n\trho_{n,i} \right] \leq -\frac{\gamma^2}{2V_{\min}}.
\end{align}
Taking $\gamma \uparrow 1$ concludes the proof.
\end{proof}

We first state the following corollary to the UMP meta-converse for DMCs.
\begin{corollary}[UMP meta converse for DMC]
\label{cor:MDconv}
For  $P_0\in\calP_n$ let $x_{P_0}^n \in \calT_{P_0}$ be an arbitrary member of type class of $P_0$ and define
\begin{align}
Q^n_{P_0,Y}(y^n) = \prod_{j=1}^n P_0 W(y_j).\end{align}
Then, any $\Mepsilon$-UMP code over DMC $W^n$ must satisfy
\begin{align}
&\epsilon_{i}\geq  \min_{P_0 \in \calP_n}\bbP \left[ \log \frac{W(Y^n|x_{P_0}^n)}{Q^n_{P_0,Y}(Y^n)}< \tau_i\right] -  \exp \left\{ \tau_i - \log M_{i} - \log \frac{1}{\lambda_{i}} + |\calA| \log(n+1)\right\}, \quad  \forall \tau_i > 0
\end{align}
for some $\blambda \in \calL_m$.
\end{corollary} 
\begin{proof}
Consider an $\Mepsilon$-UMP code and pick $P_0 \in \calP_n$.
Let $M_{P_0, i}$ be the size of constant composition component of $i$th class with empirical distribution $P_0$. Observe that the value of $\beta_{1-\epsilon_{i}} (W(\cdot|x_{P_0}^n)|| Q^n_{P_0,Y})$ is the same for all sequences $x_{P_0}^n$ in $\calT_{P_0}$. Thus, we know from Corollary~\ref{cor:constant} that
\begin{align}
M_{P_0, i}\beta_{1-\epsilon_{i}} (W(\cdot|x_{P_0}^n)|| Q^n_{P_0,Y}) \leq \lambda_{P_0,i}  \label{eq:MDbeta1}
\end{align}
for some $(\lambda_{P_0,1}, \dots, \lambda_{P_0, m})\in \calL_{m}$. 

\cite[Equation (2.67)]{thesis:Polyanskiy} states that
\begin{align}
\beta_\alpha (P,Q) \geq \frac{1}{\tau} \left( \alpha - P \left [ \frac{dP}{dQ} \geq \tau \right]\right)  \label{eq:MDbeta2}
\end{align}
for an arbitrary $\tau >0$. 

For a fixed class $i$ we combine~(\ref{eq:MDbeta1}) and~(\ref{eq:MDbeta2}) to get 
\begin{align}
\epsilon_{P_0, i} \geq \bbP \left[ \log \frac{W(Y^n|x_{P_0}^n)}{Q^n_{P_0,Y}(Y^n)}< \tau_i\right] - \exp \left\{ \tau_i - \log M_{P_0,i} - \log \frac{1}{\lambda_{P_0,i}}\right\}
\end{align}
where $\epsilon_{P_0,i}$ is the average probability of error of constant composition sub-code of class $i$. Thus the average error for the $i$th class is, 
\begin{align}
&\epsilon_{i} = \sum_{P_0 \in \calP_n}  \frac{M_{P_0,i}}{M_{i}}\epsilon_{P_0,i} \\
&\geq  \sum_{P_0 \in \calP_n}  \frac{M_{P_0,i}}{M_{i}}\bbP \left[ \log \frac{W(Y^n|x_{P_0}^n)}{Q^n_{P_0,Y}(Y^n)}< \tau_i \right] - \sum_{P_0 \in \calP_n}  \frac{M_{P_0,i}}{M_{i}} \exp \left\{ \tau_i - \log M_{P_0,i} - \log \frac{1}{\lambda_{P_0,i}}\right\}\\
&\geq  \sum_{P_0 \in \calP_n}  \frac{M_{P_0,i}}{M_{i}}\bbP \left[ \log \frac{W(Y^n|x_{P_0}^n)}{Q^n_{P_0,Y}(Y^n)}< \tau_i\right] -  \exp \left\{ \tau_i - \log M_{i} - \log \frac{1}{\lambda_{i}} + |\calA| \log(n+1)\right\}\\
&\geq  \min_{P_0 \in \calP_n}\bbP \left[ \log \frac{W(Y^n|x_{P_0}^n)}{Q^n_{P_0,Y}(Y^n)}< \tau_i\right] -  \exp \left\{ \tau_i - \log M_{i} - \log \frac{1}{\lambda_{i}} + |\calA| \log(n+1)\right\}
\end{align}
where $\lambda_{i} = \frac{1}{|\calP_n|} \sum_{P_0 \in \calP_n} \lambda_{P_0, i}$.
\end{proof}

\begin{proof}[Converse Proof of Theorem~\ref{thm:md}]
To prove the claim for a sequence of UMP codes we apply Corollary~\ref{cor:MDconv} for each $n$ to get 
\begin{align}
&\epsilon_{n,i}\geq  \min_{P_0 \in \calP_n}\bbP \left[ \log \frac{W(Y^n|x_{P_0}^n)}{Q^n_{P_0,Y}(Y^n)}< \tau_{n,i}\right] -  \exp \left\{ \tau_{n,i} - \log M_{n, i} - \log \frac{1}{\lambda_{n,i}} + |\calA| \log(n+1)\right\}, \quad  \forall \tau_{n,i} > 0 \label{eq:epsLB}
\end{align}
for all $1\leq i \leq m_n$ and some $(\lambda_{n,1}, \dots, \lambda_{n, m_n}) \in \calL_{m_n}$. 

Now we defined $\Lambda \in \calL$  by
\begin{align}
\Lambda_{n,i} &= \lambda_{n,i} \mbox{ if } 1\leq i \leq m_n\\
 \Lambda_{n,i} & =0 \mbox{ otherwise}.
\end{align}

Next, consider classes $i$ for which $\left((\rho_{n,i})_{n=1}^\infty, \left(\frac{1}{n}\log\frac{1}{\Lambda_{n,i}}\right)_{n=1}^\infty \right)$ satisfy the moderate deviations regularity conditions and for convenience define
\begin{align}
\trho_{n,i} = \rho_{n,i} - \frac{1}{n}\log \frac{1}{\Lambda_{n,i}}.
\end{align}

Pick arbitrary $\gamma>1$ and define
\begin{align}
\tau_{n,i} = nC-\gamma n\trho_{n,i}.
\end{align}
From assumptions on $( \rho_{n,i} )_{n=1}^\infty$ and $(\trho_{n,i})_{n=1}^\infty$ we know that $\trho_{n,i} \to 0$ and thus $\tau_{n,i} > 0$ for sufficiently large $n$. Evaluating~(\ref{eq:epsLB}) thus yields,
\begin{align}
\epsilon_{n,i} &\geq \min_{P_0 \in \calP_n} \bbP \left[ \log \frac{W(Y^n|x_{P_0}^n)}{Q^n_{P_0,Y}(Y^n)}< nC-\gamma n\trho_{n,i} \right] - \exp \left\{-n\trho_{n,i} (\gamma - 1) + |\calA|\log (n+1) \right\} 
\end{align}

Next, let $P_{n,i}$ be the type that achieves the minimum above for a given $n$ and $i$. By compactness of $\calP$ we may assume (by passing to a subsequence if necessary) that $P_{n,i} \to P_i^\ast$. 
We can say that
\begin{align}
\log \frac{W(Y^n|x_{P_{n,i}}^n)}{Q_Y^n(Y^n)} \sim \sum_{j=1}^n Z_j
\end{align}
where $Z_j$ are independent and
\begin{align}
\sum_{j=1}^n \bbE[Z_j] &= nI(P_{n,i}, W)\\
\sum_{j=1}^n \bbV ar[Z_j] &= nV(P_{n,i}, W)\\
\sum_{j=1}^n \bbE[|Z_j - \bbE[Z_j]|^3] &= nT(P_{n,i}, W)\\
\end{align}
Thus we obtain,
\begin{align}
\epsilon_{n,i} \geq B_{n,i} - \tilB_{n,i} \label{eq:twoterms}
\end{align}
where
\begin{align}
B_{n,i} &=\bbP \left[ \sum_{j=1}^n Z_j <nC-\gamma n\trho_{n,i} \right],\\
\tilB_{n,i} &= \exp \left\{-n\trho_{n,i}(\gamma - 1) + |\calA|\log (n+1) \right\}. 
\end{align}
Observe that if $I(P_i^\ast, W) <C$  then by Chebyshev's inequality $B_{n,i}$ converges to $1$ as $n\to \infty$. Otherwise, $I(P_i^\ast, W) =C$ and by continuity of $V(P,W)$ we have
\begin{align}
V(P_{n,i},W) \to V(P_i^\ast, W) \geq V_{\min} > 0.
\end{align}
Applying Theorem~\ref{thm:Roz} yields,
\begin{align}
B_{n,i} \geq Q\left( \frac{\gamma}{\sqrt{V(P_{n,i},W)}} \sqrt{n\trho_{n,i}^2} \right)\exp \left\{  - \frac{A_1 \gamma^3}{V^3(P_{n,i}, W)} n\trho_{n,i}^3\right\} \left(1 - \frac{\gamma A_2 T(P_{n,i}, W)}{V^2(P_{n,i},W)} \trho_{n,i} \right). \label{eq:roz}
\end{align}
And so,
\begin{align}
&\liminf_{n \to \infty} \frac{1}{n \trho_{n,i}^2} \log \bbP\left[ \sum_{j=1}^n Z_j - nI(P_{n,i},W)<-\gamma n\trho_{n,i}\right] \\
&\geq \lim_{n \to \infty} \frac{1}{n \trho_{n,i}^2} \log Q\left( \frac{\gamma}{\sqrt{V(P_{n,i},W)}} \sqrt{n\trho_{n,i}^2} \right)\\ 
&= -\frac{\gamma^2}{2V(P_i^\ast, W)} \geq -\frac{\gamma^2}{2V_{\min}}
\end{align}
Finally, observe that
\begin{align}
& \frac{1}{n \trho_{n,i}^2} \log \tilB_{n,i} =- \frac{1}{\trho_{n,i}}  \left( (\gamma -1) + \frac{|\calA|\log (n+1)}{n\trho_{n,i}} \right)\\
&=- \frac{1}{\trho_{n,i}}  \left( (\gamma -1) +\frac{1}{\sqrt{n}\trho_{n,i}} \frac{|\calA|\log (n+1)}{\sqrt{n}} \right) \to -\infty.
\end{align}
And
\begin{align}
\frac{1}{n \trho_{n,i}^2} \log \frac{\tilB_{n,i}}{B_{n,i}} \to -\infty
\end{align}
implies
\begin{align}
 \frac{\tilB_{n,i}}{B_{n,i}} \to 0,
\end{align}
so the second term in~(\ref{eq:twoterms}) is asymptotically insignificant compare to the first.

To complete the argument we need to consider the classes $i$ for which $\left((\rho_{n,i})_{n=1}^\infty, \left(\frac{1}{n}\log\frac{1}{\Lambda_{n,i}}\right)_{n=1}^\infty \right)$ do not satisfy all the moderate deviations regularity conditions. By assumption of the theorem~(\ref{eq:rho}) is always satisfied. In case that~(\ref{eq:positive}) is not satisfied let
\begin{align}
\tau_{n,i} = nC-2|\calA|\log(n+1).
\end{align}
If~(\ref{eq:positive}) is satisfied but~(\ref{eq:infinity}) is not satisfied let
\begin{align}
\tau_{n,i} = nC-2|\calA|\log(n+1)-n\trho_{n,i}.
\end{align}

Repeating the argument for the regular case we obtain that the second term in~(\ref{eq:epsLB}) goes to zero for infinitely many $n$ and RHS of~(\ref{eq:roz}) will be constant and the error is bounded away from zero.
\end{proof}
\section*{Acknowldegements}
This work was supported in part by the NSF under Grant CAREER 0844539, in part by Natural Science and Engineering Research Council of Canada (NSERC) Discovery Research Grant, and in part by  NUS    startup grant  WBS R-263-000-A98-750~(FoE). The authors would like to thank Sergio Verd\'{u} for insightful discussions. The authors would also like to thank the anonymous ISIT reviewrer and Sergio Verd\'{u} for suggesting the term `unequal message protection (UMP) codes'.


\bibliographystyle{IEEEtran}
\bibliography{Master}

\end{document}